\documentclass[12pt]{amsart}
\newtheorem{thm}{Theorem}[section]

\newtheorem{lem}[thm]{Lemma}
\newtheorem{prop}[thm]{Proposition}
\newtheorem*{prob*}{Problem}

\newtheorem*{thm*}{Theorem}

\theoremstyle{definition}
\newtheorem{defn}[thm]{Definition}

\newtheorem*{defn*}{Definition}
\newtheorem{rem}[thm]{Remark}
\numberwithin{equation}{section}

\newcommand{\tx}{\tilde{x}}
\newcommand{\ty}{\tilde{y}}

\newcommand{\C}{\mathbb C}

\newcommand{\Y}{\mathbb Y}
\newcommand{\Z}{\mathbb Z}
\newcommand{\Zp}{\mathbb Z_{\geq 0}}

\newcommand{\Mf}{\mathfrak{M}}
\newcommand{\tMf}{\widetilde{\mathfrak{M}}}
\newcommand{\DC}{\mathcal{D}}
\newcommand{\HC}{\mathcal{H}}

\DeclareMathOperator{\Conf}{Conf}
\DeclareMathOperator{\Meixner}{Meixner}

\DeclareMathOperator{\const}{const}

\DeclareMathOperator{\Prob}{Prob} 
 
\DeclareMathOperator{\Pf}{Pf} 
\DeclareMathOperator{\diag}{diag}

\begin{document}
\title[The $z$-measures on partitions and Pfaffian processes]
 {\bf{The $z$-measures on partitions, Pfaffian point processes, and  the matrix
hypergeometric kernel}}

\author{Eugene Strahov}
\address{Department of Mathematics, The Hebrew University of
Jerusalem, Givat Ram, Jerusalem
91904}\email{strahov@math.huji.ac.il}

\begin{abstract}
We consider a point process on one-dimensional lattice originated
from the harmonic analysis on the infinite symmetric group, and
defined by the $z$-measures with the deformation (Jack) parameter
2. We derive an exact Pfaffian formula for the correlation
function of this process. Namely, we prove that the correlation
function is given as a Pfaffian with a $2\times 2$ matrix kernel.
The kernel is given in terms of the Gauss hypergeometric
functions, and can be considered as a matrix analogue of the
Hypergeometric kernel introduced by A. Borodin and G. Olshanski
\cite{Borodin-4}. Our  result holds for all values of admissible
complex parameters.
$$
$$
\end{abstract}
\thanks{
Supported by US-Israel Binational Science Foundation (BSF) Grant
No. 2006333,
 and by Israel Science Foundation (ISF) Grant No. 0397937.\\
}
\maketitle \tableofcontents
\section{Introduction}
It is well-known that determinantal point processes appear in
different areas of mathematical physics, probability theory and
statistical mechanics. The theory of random Hermitian matrices
(see, for example, Deift \cite{deift}), random growth models
(Johansson \cite{johansson,johansson1}), the theory of random
power series (Peres and Vir$\acute{\mbox{a}}$g \cite{peres}) are
among numerous topics of current research where the main problems
are reduced to investigation of determinantal point processes. We
refer the reader to surveys by Soshnikov \cite{soshnikov}, and by
Hough, Krishnapur,  Peres, and Vir$\acute{\mbox{a}}$g \cite{hough}
for definitions and for different properties of determinatal point
processes.

Representation theory and the harmonic analysis on the infinite
symmetric and the infinite-dimensional unitary groups is yet
another area of mathematics where the determinantal processes play
a crucial role. The relation between determinantal processes and
representation theory of such groups was discovered by Borodin and
Olshanski in the series of papers, see Refs.
\cite{Borodin-1,Borodin-2,Borodin-3,olshanski2003}. Let us briefly
describe this relation.

 Let $S(\infty)$ denote the group
whose elements are finite permutations of $\{1,2,3,\ldots\}$. The
group $S(\infty)$ is called the infinite symmetric group, and it
is a model example of a "big" group. Set
$$ G=S(\infty)\times S(\infty),
$$
$$
K=\diag S(\infty)=\left\{(g,g)\in G \mid g\in
S(\infty)\right\}\subset G.
$$
Then $(G,K)$ is an infinite dimensional Gelfand pair in the sense
of Olshanski \cite{olshanskiGelfandPairs}. It can be shown that
the biregular spherical representation of $(G,K)$ in the space
$\ell^2\left(S(\infty)\right)$ is irreducible. Thus the
conventional scheme of noncommutative harmonic analysis is not
applicable to the case of the infinite symmetric group.

In 1993, Kerov, Olshanski and Vershik \cite{KOV1} (Kerov,
Olshanski and Vershik \cite{KOV2} contains the details)
 constructed a family $\{T_z: z\in\C\}$ of unitary representations
 of the bisymmetric infinite group $G=S(\infty)\times
 S(\infty)$. Each representation $T_z$ acts in the Hilbert space
 $L^2(\mathfrak{S},\mu_t)$, where $\mathfrak{S}$ is a certain compact space called the space of
 virtual permutations, and $\mu_t$ is a distinguished $G$-invariant probability measure on
 $\mathfrak{S}$ (here $t=|z|^2$). The representations $T_z$
 (called the generalized regular representations) are reducible.
 Moreover, it is possible to extend the definition of $T_z$ to the
 limit values $z=0$ and $z=\infty$, and it turns out that
 $T_{\infty}$ is equivalent to the biregular representation
 of $S(\infty)\times S(\infty)$. Thus, the family $\{T_z\}$ can be
 viewed as a deformation of the biregular representation.
 Once the representations $T_z$
 are constructed, the main problem of the harmonic analysis on the
 infinite symmetric group is in decomposition of the generalized
 regular representations $T_z$ into irreducible ones.

 One of the initial steps in this direction can be described as
 follows. Let $\textbf{1}$ denote the function on $\mathfrak{S}$
 identically equal to 1. Consider this function as a vector of
 $L^2(\mathfrak{S},\mu_t)$. Then $\textbf{1}$ is a spherical
 vector, and the pair $(T_z,\textbf{1})$ is a spherical
 representation of the pair $(G,K)$, see, for example, Olshanski
 \cite{olshanski2003}, Section 2. The spherical function of
 $(T_z,\textbf{1})$ is the matrix coefficient
 $(T_z(g_1,g_2)\textbf{1},\textbf{1})$, where $(g_1,g_2)\in
 S(\infty)\times S(\infty)$. Set
 $$
 \chi_z(g)=\left(T_z(g,e)\textbf{1},\textbf{1}\right),\; g\in
 S(\infty).
 $$
 The function $\chi_z$ can be understood as a character of the
 group $S(\infty)$ corresponding to $T_z$. Kerov, Olshanski and
 Vershik \cite{KOV1,KOV2} found  the restriction
 of
 $\chi_z$ to $S(n)$ in terms of irreducible characters of $S(n)$.
 Namely, let $\Y_n$ be the set of Young diagrams with $n$ boxes.
 For $\lambda\in\Y_n$ denote by $\chi^{\lambda}$ the corresponding
 irreducible character of the symmetric group $S(n)$ of degree
 $n$. Then for any $n=1,2,\ldots$ the following formula holds true
 \begin{equation}\label{EquationHiDecomposition}
 \chi_z\biggl|_{S(n)}=\sum\limits_{\lambda\in\Y_n}M^{(n)}_{z,\bar{z}}(\lambda)\frac{\chi^{\lambda}}{\chi^{\lambda}(e)}.
 \end{equation}
In this formula $M^{(n)}_{z,\bar{z}}$ is a probability measure
(called the $z$-measure) on the set of Young diagrams with $n$
boxes, or on the set of integer partitions of $n$. Formula
(\ref{EquationHiDecomposition}) defines the $z$-measure
$M^{(n)}_{z,\bar{z}}$ as a weight attached to the corresponding
Young diagram in the decomposition of the restriction of $\chi_z$
to $S(n)$ in irreducible characters of $S(n)$. Expression
(\ref{EquationHiDecomposition}) enables to reduce the problem of
decomposition of $T_z$ into irreducible components to the problem
on the computation of spectral counterparts of
$M^{(n)}_{z,\bar{z}}$.

Using a distribution on $\{0,1,2,\ldots\}$ defined by
$$
\Prob\{n\}=(1-\xi)^{z\bar z}\frac{(z\bar z)_n}{n!}\xi^n,\;\;\xi>0
$$
(where $(a)_n$ stands for $a(a+1)\ldots (a+n-1)$) it is possible
to mix distributions $M^{(n)}_{z,\bar{z}}$, and to obtain a
distribution $M_{z,\bar{z},\xi}$ on the set of all Young diagrams.

 It was shown by Borodin
and Olshanski in Ref. \cite{Borodin-4}, that $M_{z,\bar{z},\xi}$
defines a determinantal point process on one-dimensional lattice.
The kernel of this process has the integrable form in the sense of
Its, Izergin, Korepin, and Slavnov \cite{its}, and can be written
in terms of the Gauss hypergeometric functions. This fact was
proved in many ways in a variety of papers (see, for example,
Okounkov \cite{okounkov5}, Borodin, Olshanski, and Strahov
\cite{borodin4}, Borodin and Olshanski \cite{BO}, and references
therein). The relation between representation theory of big groups
and determinantal point processes  gave rise to numerous
applications from enumerative combinatorics and random growth
models to the theory of Painlev$\acute{\mbox{e}}$ equations, see
Borodin and Deift \cite{BorodinDeift}.

It is known that if $z, z'\rightarrow \infty$ and
$\xi=\frac{\eta}{zz'}\rightarrow 0$, where $\eta>0$ is fixed, then
$M_{z,z',\xi}$ tends to Poissonized Plancherel distribution
studied in many papers (see, for example, Baik, Deift and
Johansson \cite{baik}). In particular, it was demonstrated that
the Poissonized Plancherel distribution is  similar to the
Gaussian unitary ensemble (GUE) of random matrix theory, which is
an example of an ensemble from the $\beta=2$ symmetry class. On
the other hand, in addition to ensembles of $\beta=2$ symmetry
class, random matrix theory deals with ensembles of $\beta=1$ and
$\beta=4$ symmetry classes.  Note that ensembles from both
$\beta=1$ and $\beta=4$ symmetry classes (in contrast to those
from $\beta=2$ symmetry class) lead to Pfaffian point processes,
and analogy between random partitions and random matrices
naturally motivates a search for Pfaffian point processes
originated from the representation theory of the infinite
symmetric group.

It is the  purpose of the present paper to construct and to
investigate Pfaffian point processes  relevant for the
representation theory and for the harmonic analysis on the
infinite symmetric group. It turns out  that such processes are
determined by $z$-measures with the Jack parameters $\theta=2$ and
$\theta=1/2$. The fact that these measures play a role in the
harmonic analysis was established by Olshanski
\cite{olshanskiletter}, and the detailed explanation of this
representation-theoretic aspect can be found in Strahov
\cite{strahov1}.  Due to the fact that $z$-measures with the Jack
parameters $\theta=2$ and $\theta=1/2$ are related to each other
in a very simple way (see Proposition
\ref{PropositionMSymmetries}), it is enough to consider a point
process defined by the $z$-measure with the Jack parameter
$\theta=2$. The main new result of the present paper is in
explicit computation of the correlation functions for this
measure.  We prove that the correlation functions of the processes
are given by Pfaffian formulas  with $2\times 2$ matrix valued
kernel. The kernel is constructed in terms of the Gauss
hypergeometric functions. Our  result holds for all values of
admissible complex parameters $z, z'$.

Once the relation with the harmonic analysis on the infinite
symmetric group is the main motivation behind this work, we expect
different applications of our results in enumerative combinatorics
and statistical physics similar to the case of the $z$-measures
with the Jack parameter $\theta=1$ studied by Borodin and
Olshanski.

\textbf{Acknowledgements} I am very  grateful to Alexei Borodin
and Grigori Olshanski for numerous discussions at different stages
of this work, and for many valuable comments.
\section{Definitions and the main result}
\subsection{The $z$-measures on partitions with the general
parameter $\theta>0$}\label{Sectionztheta} We use Macdonald
\cite{macdonald} as a basic reference for the notations related to
integer partitions and to symmetric functions. In particular,
every decomposition
$$
\lambda=(\lambda_1,\lambda_2,\ldots,\lambda_l):\;
n=\lambda_1+\lambda_2+\ldots+\lambda_{l},
$$
where $\lambda_1\geq\lambda_2\geq\ldots\geq\lambda_l$ are positive
integers, is called an integer partition. We identify integer
partitions with the corresponding Young diagrams.  The set of
Young diagrams with $n$ boxes  is denoted by $\Y_n$.

Following Borodin and Olshanski \cite{BO1}, Section 1, and Kerov
\cite{kerov} let $M_{z,z',\theta}^{(n)}$ be a complex measure on
$\Y_n$ defined by
\begin{equation}\label{EquationVer4zmeasuren}
M_{z,z',\theta}^{(n)}(\lambda)=\frac{n!(z)_{\lambda,\theta}(z')_{\lambda,\theta}}{(t)_nH(\lambda,\theta)H'(\lambda,\theta)},
\end{equation}
where $n=1,2,\ldots $, and where we use the following notation
\begin{itemize}
    \item $z,z'\in\C$ and $\theta>0$ are parameters, the parameter
    $t$ is defined by
    $$
    t=\frac{zz'}{\theta}.
    $$
    \item $(t)_n$ stands for the Pochhammer symbol,
    $$
    (t)_n=t(t+1)\ldots (t+n-1)=\frac{\Gamma(t+n)}{\Gamma(t)}.
    $$
    \item
    $(z)_{\lambda,\theta}$ is a multidemensional analogue of the
    Pochhammer symbol defined by
    $$
    (z)_{\lambda,\theta}=\prod\limits_{(i,j)\in\lambda}(z+(j-1)-(i-1)\theta)
    =\prod\limits_{i=1}^{l(\lambda)}(z-(i-1)\theta)_{\lambda_i}.
    $$
     Here $(i,j)\in\lambda$ stands for the box in the $i$th row
     and the $j$th column of the Young diagram $\lambda$, and we
     denote by $l(\lambda)$ the number of nonempty rows in the
     Young diagram $\lambda$.
    \item
    $$
    H(\lambda,\theta)=\prod\limits_{(i,j)\in\lambda}\left((\lambda_i-j)+(\lambda_j'-i)\theta+1\right),
   $$
   $$
     H'(\lambda,\theta)=\prod\limits_{(i,j)\in\lambda}\left((\lambda_i-j)+(\lambda_j'-i)\theta+\theta\right),
   $$
      where $\lambda'$ denotes the transposed diagram.
\end{itemize}
\begin{prop}\label{PropositionHH}
The following symmetry relations hold true
$$
H(\lambda,\theta)=\theta^{|\lambda|}H'(\lambda',\frac{1}{\theta}),\;\;(z)_{\lambda,\theta}
=(-\theta)^{|\lambda|}\left(-\frac{z}{\theta}\right)_{\lambda',\frac{1}{\theta}}.
$$
Here $|\lambda|$ stands for the number of boxes in the diagram
$\lambda$.
\end{prop}
\begin{proof}
These relations follow immediately from definitions of
$H(\lambda,\theta)$ and $(z)_{\lambda,\theta}$.
\end{proof}
\begin{prop}\label{PropositionMSymmetries}
We have
$$
M_{z,z',\theta}^{(n)}(\lambda)=M_{-z/\theta,-z'/\theta,1/\theta}^{(n)}(\lambda').
$$
\end{prop}
\begin{proof}
Use definition of $M_{z,z',\theta}^{(n)}(\lambda)$, equation
(\ref{EquationVer4zmeasuren}), and apply Proposition
\ref{PropositionHH}.
\end{proof}
\begin{prop}\label{Prop1.3}
We have
$$
\sum\limits_{\lambda\in\Y_n}M_{z,z',\theta}^{(n)}(\lambda)=1.
$$
\end{prop}
\begin{proof}
See Kerov \cite{kerov}, Borodin and Olshanski
\cite{BO1,BOHARMONICFUNCTIONS}.
\end{proof}
\begin{prop}\label{PropositionSeries}
If parameters $z, z'$ satisfy one of the three conditions listed
below, then the measure $M_{z,z',\theta}^{(n)}$ defined by
expression (\ref{EquationVer4zmeasuren}) is a probability measure
on $Y_n$. The conditions are as follows.\begin{itemize}
    \item Principal series: either
$z\in\C\setminus(\Z_{\leq 0}+\Zp\theta)$ and $z'=\bar z$.
    \item The complementary series: the parameter $\theta$ is a rational number, and both $z,z'$
are real numbers lying in one of the intervals between two
consecutive numbers from the lattice $\Z+\Z\theta$.
    \item The degenerate series: $z,z'$ satisfy one of the
    following conditions\\
    (1) $(z=m\theta, z'>(m-1)\theta)$ or $(z'=m\theta,
    z>(m-1)\theta)$;\\
    (2) $(z=-m, z'<-m+1)$ or $(z'=-m,
    z<m-1)$.
\end{itemize}

\end{prop}
\begin{proof} See Propositions 1.2, 1.3 in Borodin and Olshanski
\cite{BO1}.
\end{proof}
Thus, if the  conditions in the Proposition above  are satisfied,
then $M_{z,z',\theta}^{(n)}$ is a probability measure defined on
$\Y_n$, as  follows from Proposition \ref{Prop1.3}.
\begin{rem}
When both $z,z'$ go to infinity, expression
(\ref{EquationVer4zmeasuren}) has a limit
\begin{equation}\label{EquationPlancherelInfy}
M_{\infty,\infty,\theta}^{(n)}(\lambda)=\frac{n!\theta^{n}}{H(\lambda,\theta)H'(\lambda,\theta)}
\end{equation}
called the Plancherel measure on $\Y_n$ with general $\theta>0$.
Statistics of the Plancherel measure with the general Jack
parameter $\theta>0$ is discussed in  many papers, see, for
example, a very recent paper by Matsumoto \cite{matsumoto}, and
references therein. Matsumoto \cite{matsumoto} compares limiting
distributions of rows of random partitions with distributions of
certain random variables from a traceless Gaussian
$\beta$-ensemble.
\end{rem}
It is convenient  to mix all measures $M_{z,z',\theta}^{(n)}$, and
to define a new measure $M_{z,z',\xi,\theta}$  on
$\Y=\Y_0\cup\Y_1\cup\ldots $. Namely, let $\xi\in(0,1)$ be an
additional parameter, and set
\begin{equation}\label{EquationMzztheta}
M_{z,z',\xi,\theta}(\lambda)=(1-\xi)^t\xi^{|\lambda|}
\frac{(z)_{\lambda,\theta}(z')_{\lambda,\theta}}{H(\lambda,\theta)H'(\lambda,\theta)}.
\end{equation}
\begin{prop} We have
$$
\sum\limits_{\lambda\in\Y}M_{z,z',\xi,\theta}(\lambda)=1.
$$
\end{prop}
\begin{proof}
Follows immediately from Proposition \ref{Prop1.3}.
\end{proof}
If conditions on $z,z'$ formulated in Propositions 1.2, 1.3 in
Borodin and Olshanski \cite{BO} are satisfied, then
$M_{z,z',\xi,\theta}(\lambda)$ is a probability measure on $\Y$.
We will refer to $M_{z,z',\xi,\theta}(\lambda)$ as to the
$z$-measure with the deformation (Jack) parameter $\theta$.
\subsection{A basis in the $l^2$ space on the lattice $\Z'=\Z+\frac{1}{2}$}

In this Section we describe a basis in the $l^2$ space on the
$1$-dimensional lattice $\Z'=\Z+\frac{1}{2}$ introduced in Borodin
and Olshanski \cite{BO}, and define certain operators acting in
the space $l^2$.  Elements of $\Z'=\Z+\frac{1}{2}$ will be denoted
by $x, y$. Introduce the principal series, the complementary
series, and the degenerate series as in Proposition
\ref{PropositionSeries} with $\theta=1$.  Assume that parameters
$z$, $z'$ are in the principal series or in the complementary
series, but not in the degenerate series. Therefore, the
conditions on $z, z'$ are as follows.
\begin{itemize}
    \item The numbers $z, z'$ are not real and are conjugate to
    each other (principal series).
    \item Both $z, z'$ are real and are contained in the same open
    interval of the form $(m,m+1)$, where $m\in\Z$.
\end{itemize}

In this case we say that parameters $z, z'$ are admissible. In
particular, $z$,$z'$ are not integers. Introduce a family of
functions on $\Z'=\Z+\frac{1}{2}$ depending on a parameter
$a\in\Z'$, and also on the parameters $z,z',\xi$:
\begin{equation}\label{PsiaF}
\begin{split}
\psi_a(x;z,z',\xi)=\left(\frac{\Gamma(x+z+\frac{1}{2})\Gamma(x+z'+\frac{1}{2})}{\Gamma(z-a+\frac{1}{2})\Gamma(z'-a+\frac{1}{2})}\right)^{1/2}
\xi^{\frac{x+a}{2}}(1-\xi)^{\frac{z+z'}{2}-a}\\
\times\frac{F\left(-z+a+\frac{1}{2},-z'+a+\frac{1}{2};
x+a+1;\frac{\xi}{\xi-1}\right)}{\Gamma(x+a+1)},\; x\in\Z',
\end{split}
\end{equation}
where $F(A,B;C;w)$ is the Gauss hypergeometric function. As it is
explained in Borodin and Olshanski \cite{BO}, Section 2, the above
expression makes sense, and the functions $\psi_a(x;z,z',\xi)$ are
real-valued. In particular, the assumptions on $(z,z')$ imply that
$\Gamma(x+z+\frac{1}{2})$ and $\Gamma(x+z'+\frac{1}{2})$ have no
singularities for $x\in\Z'$, and that
$$
\Gamma(x+z+\frac{1}{2})\Gamma(x+z'+\frac{1}{2})>0,\;\;\Gamma(z-a+\frac{1}{2})\Gamma(z'-a+\frac{1}{2})>0.
$$
so we can take the positive values of the square roots in equation
(\ref{PsiaF}).
\begin{prop} a) Introduce a second order difference operator
$D(z,z',\xi)$ on the lattice $\Z'$, depending on parameters
$z,z',\xi$ and acting on functions $f(x)$ (where $x$ ranges over
$\Z'$) as follows
\begin{equation}
\begin{split}
D(z,z',\xi)f(x)&=\sqrt{\xi(z+x+\frac{1}{2})(z'+x+\frac{1}{2})}f(x+1)\\
&+\sqrt{\xi(z+x+\frac{1}{2})(z'+x+\frac{1}{2})}f(x-1)-(x+\xi(z+z'+x))f(x).
\end{split}
\nonumber
\end{equation}
Then the functions $\psi_a(x;z,z',\xi)$, where $a$ ranges over
$\Z'$, are eigenvalues of the operator $D(z,z',\xi)$,
$$
D(z,z',\xi)\psi_a(x;z,z',\xi)=a(1-\xi)\psi_a(x;z,z',\xi).
$$
b) The functions $\psi_a(x;z,z',\xi)$, where $a$ ranges over
$\Z'$, form an orthonormal basis in the Hilbert space $l^2(\Z')$.
\end{prop}
\begin{proof}
See  Borodin and Olshanski \cite{BO}, Section 2.
\end{proof}
\begin{prop}\label{LemmaFASCountourIntegral}
For any $A, B\in\C, M\in\Z$, and $\xi\in(0,1)$ we have
\begin{equation}
\begin{split}
\frac{1}{2\pi
i}\oint\limits_{\{w\}}(1-\sqrt{\xi}w)^{A-1}(1-\frac{\sqrt{\xi}}{w})^{-B}
\frac{dw}{w^{M+1}}=\xi^{\frac{M}{2}}(1-\xi)^{-B}\frac{\Gamma(-A+M+1)}{\Gamma(-A+1)\Gamma(M+1)}F(A,B;M+1;\frac{\xi}{\xi-1}).
\end{split}
\nonumber
\end{equation}
Here $\xi\in(0,1)$ and $\{w\}$ is an arbitrary simple contour
which goes around the points $0$ and $\xi$ in the positive
direction leaving $1/\sqrt{\xi}$ outside.
\end{prop}
\begin{proof}
See Borodin and Olshanski \cite{BO}, Lemma 2.2.
\end{proof}
\begin{prop}\label{PSIa}
We have the following integral representations
\begin{equation}
\begin{split}
\psi_a(x;z,z',\xi)=&\left(\frac{\Gamma(x+z+\frac{1}{2})\Gamma(x+z'+\frac{1}{2})}{\Gamma(z-a+\frac{1}{2})\Gamma(z'-a+\frac{1}{2})}\right)^{1/2}
\frac{\Gamma(z'-a+\frac{1}{2})}{\Gamma(x+z'+\frac{1}{2})}
(1-\xi)^{\frac{z'-z+1}{2}}\\
&\times\frac{1}{2\pi
i}\oint\limits_{\{w\}}(1-\sqrt{\xi}w)^{-z'+a-\frac{1}{2}}(1-\frac{\sqrt{\xi}}{w})^{z-a-\frac{1}{2}}w^{-x-a}\frac{dw}{w},
\end{split}
\nonumber
\end{equation}
where $\{w\}$ is an arbitrary simple loop, oriented in positive
direction, surrounding the points $0$ and $\sqrt{\xi}$, and
leaving $1/\sqrt{\xi}$ outside.
\end{prop}
\begin{proof}
Follows immediately from equation (\ref{PsiaF}), and from
Proposition \ref{LemmaFASCountourIntegral}.
\end{proof}
Let $\underline{K}_{z,z',\xi}$ be the orthogonal projection
operator in $l^2(\Z')$ whose range is the subspace spanned by the
basis vectors $\psi_a$ with indexes $a\in\Z'_+\subset\Z'$. If
$\underline{K}_{z,z',\xi}(x,y)$ is the matrix of
$\underline{K}_{z,z',\xi}$, then
\begin{equation}\label{5.6112}
\underline{K}_{z,z',\xi}(x,y)=\sum\limits_{a\in\Zp+\frac{1}{2}}
\psi_a(x;z,z',\xi)\psi_a(y;z,z',\xi).
\end{equation}
\begin{prop}\label{EquationRepresentationKzz'} The function $\underline{K}_{z,z',\xi}(x,y)$ can be
written in the form
\begin{equation}
\begin{split}
&\underline{K}_{z,z',\xi}(x,y)=\frac{1}{(2\pi
i)^2}\sqrt{\frac{\Gamma(x+z+\frac{1}{2})\Gamma(y+z'+\frac{1}{2})}{\Gamma(x+z'+\frac{1}{2})\Gamma(y+z+\frac{1}{2})}}\\
&\times\oint\limits_{\{w_1\}}\oint\limits_{\{w_2\}}
\frac{(1-\sqrt{\xi}w_1)^{-z'}(1-\frac{\sqrt{\xi}}{w_1})^{z}
(1-\sqrt{\xi}w_2)^{-z}(1-\frac{\sqrt{\xi}}{w_2})^{z'}}{w_1w_2-1}
\frac{dw_1}{w_1^{x+\frac{1}{2}}}\frac{dw_2}{w_2^{y+\frac{1}{2}}}.
\end{split}
\nonumber
\end{equation}
where $\{w_1\}$ and $\{w_2\}$ are arbitrary simple contours
satisfying the following conditions
\begin{itemize}
    \item  both contours go around $0$ in positive direction;
    \item  the point $\xi^{1/2}$ is in the interior of each of the
    contours while the point $\xi^{-1/2}$ lies outside the
    contours;
    \item the contour $\{w_1^{-1}\}$ is contained in the interior
    of the contour $\{w_2\}$ (equivalently, $\{w_2^{-1}\}$ is
    contained in the interior of $\{w_1\}$).
\end{itemize}
\end{prop}
\begin{proof}
See Borodin and Olshanski \cite{BO}, Theorem 3.3.
\end{proof}
\subsection{The main result: the formula for the correlation function of the $z$-measure with the
Jack parameter $\theta=2$. The matrix hypergeometric kernel} By a
point configuration in $\Z'$ we mean any subset of $\Z'$. Let
$\Conf(\Z')$ be the set of all point configurations, and assume
that we are given a probability measure on $\Conf(\Z')$. Then we
can speak about random point configurations in $\Conf(\Z')$. The
$n$th point correlation function of the given probability measure
is defined by
$$
\varrho_n(x_1,\ldots,x_n)=\Prob\{\mbox{the random configuration
contains $x_1,\ldots, x_n$}\}.
$$
Here $n=1,2,\ldots ,$ and $x_1,\ldots ,x_n$ are pairwise distinct
points in $\Z'$.

We say that a given probability measure defines a Pfaffian point
process on $\Conf(\Z')$ if there exists a $2\times 2$ matrix
valued kernel $\mathbb{K}(x,y)$ on $\Z'\times\Z'$ such that
$$
\varrho_n(x_1,\ldots,x_n)=\Pf\left[\mathbb{K}(x_i,x_j)\right]_{i,j=1}^n,\;\;n=1,2,\ldots
.
$$
The kernel $\mathbb{K}(x,y)$ is referred to as the correlation
kernel of the Pfaffian point process under considerations.

 Set
$\DC_{2}(\lambda)=\left\{\lambda_i-2i+\frac{1}{2}\right\}.$ Thus
$\DC_{2}(\lambda)$ is an infinite subset of $\Z'$ corresponding to
the Young diagram $\lambda$. Let $X=(x_1,\ldots,x_n)$ be a subset
of $\Z'$ consisting of $n$ pairwise distinct points, and define
$$
\varrho^{(z,z',\xi,\theta=2)}_n(x_1,\ldots,x_n)=M_{z,z',\xi,\theta=2}\left(\{\lambda|X\subset\DC_2(\lambda)\}\right).
$$
If $M_{z,z',\xi,\theta=2}$ is positive, then it is a probability
measure defined on $\Y$, and
$\varrho^{(z,z',\xi,\theta)}(x_1,\ldots,x_n)$ is the probability
that the random point configuration $\DC_{2}(\lambda)$ contains
the fixed $n$-point configuration $X=(x_1,\ldots,x_n)$. The
function $\varrho^{(z,z',\xi,\theta=2)}_n(x_1,\ldots,x_n)$ can be
understood as the correlation function of the point process
defined by the measure $M_{z,z',\xi,\theta=2}$.

The main result of the present paper is in explicit computation of
$\varrho^{(z,z',\xi,\theta)}(x_1,\ldots,x_n)$ for admissible
parameters $z$ and $z'$, for which both the measure
$M_{z,z',\xi,\theta=2}$ is positive, and the functions
$\psi_{a}(x;z,z',\xi)$, $\underline{K}_{z,z',\xi}(x;y)$ are well
defined by equations (\ref{PsiaF}), (\ref{5.6112})
correspondingly. In order to present our result, let us introduce
the functions $\left(E(z,z')\underline{K}_{z,z',\xi}\right)(x;y)$,
$\left(E(z,z')\psi_{-\frac{1}{2}}\right)(x;z,z',\xi)$, and
$\left(E(z,z')\psi_{\frac{1}{2}}\right)(x;z,z',\xi)$. We first
define these functions in terms of infinite series. Namely, for
any admissible $z, z'$ (i.e. for $z, z'$ from the principal or
complementary series defined as in Proposition
\ref{PropositionSeries} with $\theta=1$) these functions are given
by the following formulae. If $x-\frac{1}{2}$ is an even integer,
then
\begin{equation}
\left(E(z,z')\underline{K}_{z,z',\xi}\right)(x;y)=-\sum\limits_{l=0}^{\infty}
    \sqrt{\frac{\left(x+z+\frac{3}{2}\right)_{l,2}\left(x+z'+\frac{3}{2}\right)_{l,2}}{\left(x+z+\frac{5}{2}\right)_{l,2}\left(x+z'+\frac{5}{2}\right)_{l,2}}}
    \;\;\underline{K}_{z,z',\xi}(x+2l+1;y),
    \nonumber
\end{equation}
and
\begin{equation}
\left(E(z,z')\psi_{\pm
\frac{1}{2}}\right)(x;z,z',\xi)=-\sum\limits_{l=0}^{\infty}
    \sqrt{\frac{\left(x+z+\frac{3}{2}\right)_{l,2}\left(x+z'+\frac{3}{2}\right)_{l,2}}{\left(x+z+\frac{5}{2}\right)_{l,2}\left(x+z'+\frac{5}{2}\right)_{l,2}}}
    \;\;\psi_{\pm \frac{1}{2}}(x+2l+1;z,z',\xi).
    \nonumber
\end{equation}
Otherwise, if $x-\frac{1}{2}$ is an odd integer, then
$\left(E(z,z')\underline{K}_{z,z',\xi}\right)(x;y)$,
$\left(E(z,z')\psi_{-\frac{1}{2}}\right)(x;z,z',\xi)$, and
$\left(E(z,z')\psi_{\frac{1}{2}}\right)(x;z,z',\xi)$ are defined
by
\begin{equation}
\left(E(z,z')\underline{K}_{z,z',\xi}\right)(x,y)=
\sum\limits_{l=1}^{\infty}
    \sqrt{\frac{\left(-x-z-\frac{1}{2}\right)_{l,2}\left(-x-z'-\frac{1}{2}\right)_{l,2}}{\left(-x-z+\frac{1}{2}\right)_{l,2}\left(-x-z'+\frac{1}{2}\right)_{l,2}}}
    \;\;\underline{K}_{z,z',\xi}(x-2l+1;y),
    \nonumber
\end{equation}
and
\begin{equation}
\left(E(z,z')\psi_{\pm
\frac{1}{2}}\right)(x;z,z',\xi)=\sum\limits_{l=1}^{\infty}
\sqrt{\frac{\left(-x-z-\frac{1}{2}\right)_{l,2}\left(-x-z'-\frac{1}{2}\right)_{l,2}}{\left(-x-z+\frac{1}{2}\right)_{l,2}\left(-x-z'+\frac{1}{2}\right)_{l,2}}}
\;\;\psi_{\pm \frac{1}{2}}(x-2l+1;z,z',\xi).
    \nonumber
\end{equation}
Here $(x)_{n,k}$ denotes the Pochhammer $k$-symbol,
$$
(x)_{n,k}=x(x+k)(x+2k)\ldots (x+(n-1)k).
$$
Let us explain why the formulae above make sense. Once the
parameters $z, z'$ are admissible, all the expressions inside
square roots are strictly positive, so we can take the positive
values of the square roots in equations for
$\left(E(z,z')\underline{K}_{z,z',\xi}\right)(x,y)$,
$\left(E(z,z')\psi_{-\frac{1}{2}}\right)(x;z,z',\xi)$, and
$\left(E(z,z')\psi_{\frac{1}{2}}\right)(x;z,z',\xi)$ just written
above. Using Propositions \ref{PSIa} and
\ref{EquationRepresentationKzz'} we can represent these functions
as well-defined contour integrals, finite for all $x,y\in\Z'$. For
example, if $x-\frac{1}{2}$ is an even integer, then
\begin{equation}\label{KF}
\begin{split}
\left(E(z,z')\underline{K}_{z,z',\xi}\right)(x,y)&=-\frac{1}{(2\pi
i)^2}\frac{\Gamma(x+z+\frac{3}{2})\Gamma(y+z'+\frac{1}{2})}{\sqrt{\Gamma(x+z+\frac{3}{2})
\Gamma(x+z'+\frac{3}{2})\Gamma(y+z+\frac{1}{2})\Gamma(y+z'+\frac{1}{2})}}\\
&\times\oint\limits_{\{w_1\}}\oint\limits_{\{w_2\}}
\frac{(1-\sqrt{\xi}w_1)^{-z'}(1-\frac{\sqrt{\xi}}{w_1})^{z}
(1-\sqrt{\xi}w_2)^{-z}(1-\frac{\sqrt{\xi}}{w_2})^{z'}}{w_1w_2-1}\\
&\times
F\left(\frac{x+z+\frac{3}{2}}{2},1;\frac{x+z'+\frac{5}{2}}{2};\frac{1}{w_1^2}\right)\frac{dw_1}{w_1^{x+\frac{3}{2}}}\frac{dw_2}{w_2^{y+\frac{1}{2}}}.
\end{split}
\end{equation}
where the contours $\{w_1\}, \{w_2\}$ are chosen as in Proposition
\ref{EquationRepresentationKzz'}, and are contained in the domain
$|w|>1$. In the domain $|w_1|>1$ the  Gauss hypergeometric
function inside the integral is an analytic function of $w_1$, and
can be represented by the uniformly convergent series,
$$
F\left(\frac{x+z+\frac{3}{2}}{2},1;\frac{x+z'+\frac{5}{2}}{2};\frac{1}{w_1^2}\right)=\sum\limits_{l=0}^{\infty}
\frac{\left(\frac{x+ z+\frac{3}{2}}{2}\right)_l}{\left(\frac{x+
z'+\frac{5}{2}}{2}\right)_l}\frac{1}{w_1^{2l}}.
$$
Once the contour $\{w_1\}$ is chosen in the domain where the
series above is uniformly convergent, we can interchange summation
and integration, and, expressing each integral in the sum in terms
of function $\underline{K}_{z,z',\xi}(x,y)$, we arrive to the
series in the definition of
$\left(E(z,z')\underline{K}_{z,z',\xi}\right)(x,y)$. Since the
righthand side in equation (\ref{KF}) is finite for all $x, y\in
\Z'$, we conclude that the series in the definition of
$\left(E(z,z')\underline{K}_{z,z',\xi}\right)(x,y)$ is convergent
for all $x, y\in \Z'$.

Now we are in position to formulate the main result of this work.
\begin{thm}\label{MAINTHEOREM1} For any admissible  $z, z'$ the $z$-measure with the Jack parameter $\theta=2$ defines a Pfaffian
point process. Namely, for any admissible  $z, z'$ the $n$-point
correlation function of $M_{z,z',\xi,\theta=2}(\lambda)$ is given
by
\begin{equation}\label{equation2.111}
\varrho_n^{(z,z',\xi,\theta=2)}(x_1,\ldots,x_n)=\Pf\left[\underline{\mathbb{K}}_{z,z',\xi,\theta=2}(x_i,x_j)\right]_{i,j=1}^n,
\end{equation}
where the correlation kernel
$\underline{\mathbb{K}}_{z,z',\xi,\theta=2}(x,y)$ has the
following form
$$
\underline{\mathbb{K}}_{z,z',\xi,\theta=2}(x,y)=
\left[\begin{array}{cc}
  \underline{S}_{z,z',\xi,\theta=2}(x,y) & -\underline{SD_-}_{z,z',\xi,\theta=2}(x,y) \\
  -\underline{D_+S}_{z,z',\xi,\theta=2}(x,y) & \underline{D_+SD_-}_{z,z',\xi,\theta=2}(x,y) \\
\end{array}\right].
$$
In the formula above
\begin{equation}
\begin{split}
    &\underline{S}_{z,z',\xi,\theta=2}(x,y)=\sqrt{(z+y+\frac{1}{2})(z'+y+\frac{1}{2})}
    \left(E(z,z')\underline{K}_{z,z',\xi}\right)(x;y)\\
    &+\sqrt{zz'}\left(E(z,z')\psi_{-\frac{1}{2}}\right)(x;z,z',\xi)
\left(E(z,z')\psi_{\frac{1}{2}}\right)(y;z,z',\xi),
\end{split}
\nonumber
\end{equation}
\begin{equation}
\underline{D_+S}_{z,z',\xi,\theta=2}(x,y)=\frac{1}{\sqrt{(z+x+\frac{3}{2})(z'+x+\frac{3}{2})}}
\underline{S}_{z,z',\xi,\theta=2}(x+1,y), \nonumber
\end{equation}
\begin{equation}
\underline{SD_-}_{z,z',\xi,\theta=2}(x,y)=\frac{1}{\sqrt{(z+y+\frac{3}{2})(z'+y+\frac{3}{2})}}
\underline{S}_{z,z',\xi,\theta=2}(x,y+1), \nonumber
\end{equation}
and
\begin{equation}
\underline{D_+SD_-}_{z,z',\xi,\theta=2}(x,y)=\frac{1}{\sqrt{(z+x+\frac{3}{2})(z'+x+\frac{3}{2})(z+y+\frac{3}{2})(z'+y+\frac{3}{2})}}
\underline{S}_{z,z',\xi,\theta=2}(x+1,y+1). \nonumber
\end{equation}
\end{thm}
\subsection{Remarks on Theorem \ref{MAINTHEOREM1}}
\subsubsection{}
All matrix elements of
$\underline{\mathbb{K}}_{z,z',\xi,\theta=2}(x,y)$ are constructed
in terms of the Gauss hypergeometric functions, so it is natural
to refer to $\underline{\mathbb{K}}_{z,z',\xi,\theta=2}(x,y)$ as
to the \textit{matrix hypergeometric kernel}.
\subsubsection{}
All matrix elements of
$\underline{\mathbb{K}}_{z,z',\xi,\theta=2}(x,y)$ are symmetric
with respect to $z\longleftrightarrow z'$. This implies that the
correlation function is symmetric with respect to
$z\longleftrightarrow z'$ as well. The fact that this symmetry
relation must be satisfied is evident from the symmetry of the
$z$-measure under considerations under $z\longleftrightarrow z'$.
\subsubsection{}
It is possible to present the function
$\underline{S}_{z,z',\xi,\theta=2}(x,y)$ (which defines the matrix
kernel $\underline{\mathbb{K}}_{z,z',\xi,\theta=2}(x,y)$) in a
form which is manifestly antisymmetric with respect to
$x\longleftrightarrow y$. For this purpose let us introduce the
functions $\mathcal{P}(x,w,z,z')$ and $\mathcal{Q}(x,w,z,z')$. If
$x-\frac{1}{2}$ is even, these functions are defined by
\begin{equation}\label{Equation2.3.1.1}
\begin{split}
\mathcal{P}(x,w,z,z')&=\frac{\Gamma(x+z+\frac{1}{2})}{\sqrt{\Gamma(x+z+\frac{3}{2})\Gamma(x+z'+\frac{3}{2})}}\\
&\times F\left(\frac{x+z+\frac{3}{2}}{2},
1;\frac{x+z'+\frac{5}{2}}{2};\frac{1}{w^2}\right)w^{-x-\frac{1}{2}},
\end{split}
\end{equation}
and by
\begin{equation}\label{Equation2.3.1.2}
\begin{split}
\mathcal{Q}(x,w,z,z')&=\frac{\Gamma(x+z'+\frac{3}{2})}{\sqrt{\Gamma(x+z+\frac{3}{2})\Gamma(x+z'+\frac{3}{2})}}\\
&\times \left[F\left(\frac{x+z'+\frac{3}{2}}{2},
1;\frac{x+z+\frac{1}{2}}{2};\frac{1}{w^2}\right)-1\right]w^{-x+\frac{1}{2}}.
\end{split}
\end{equation}
If $x-\frac{1}{2}$ is odd, these functions are defined by
\begin{equation}\label{Equation2.3.1.3}
\begin{split}
\mathcal{P}(x,w,z,z')&=-\frac{\Gamma(x+z+\frac{1}{2})}{\sqrt{\Gamma(x+z+\frac{3}{2})\Gamma(x+z'+\frac{3}{2})}}\\
&\times\left[F\left(-\frac{x+z'+\frac{1}{2}}{2},1;-\frac{x+z'-\frac{1}{2}}{2};w^2\right)-1\right]
w^{-x-\frac{1}{2}},
\end{split}
\end{equation}
and by
\begin{equation}\label{Equation2.3.1.4}
\begin{split}
\mathcal{Q}(x,w,z,z')&=-\frac{\Gamma(x+z'+\frac{3}{2})}{\sqrt{\Gamma(x+z+\frac{3}{2})\Gamma(x+z'+\frac{3}{2})}}\\
&\times F\left(-\frac{x+z-\frac{3}{2}}{2},
1;-\frac{x+z'-\frac{1}{2}}{2};w^2\right)w^{-x+\frac{1}{2}}.
\end{split}
\end{equation}
Set
\begin{equation}\label{Equation2.3.1.5}
\begin{split}
\widetilde{\underline{S}}_{z,z',\xi}(x,y)=\frac{1}{(2\pi i)^2}
\oint\limits_{\{w_1\}}\oint\limits_{\{w_2\}}&
\frac{(1-\sqrt{\xi}w_1)^{-z'}(1-\frac{\sqrt{\xi}}{w_1})^{z}
(1-\sqrt{\xi}w_2)^{-z}(1-\frac{\sqrt{\xi}}{w_2})^{z'}}{w_1w_2-1}
\frac{dw_1}{w_1}\frac{dw_2}{w_2}\\
&\times\mathcal{P}(x,w_1,z,z')\mathcal{Q}(y,w_2,z,z'),
\end{split}
\end{equation}
where the contours $\{w_1\},\{w_2\}$ are chosen as in the
statement of Proposition \ref{EquationRepresentationKzz'} with the
following additional conditions. If $x-\frac{1}{2}$ is an even
integer, then $\{w_1\}$ lies in the domain $|w|>1$. If
$x-\frac{1}{2}$ is an odd integer, then $\{w_1\}$ lies in the
domain $|w|<1$. The  same condition is imposed on $\{w_2\}$: if
$y-\frac{1}{2}$ is an even integer, then $\{w_2\}$ lies in the
domain $|w|>1$, and  if $x-\frac{1}{2}$ is an odd integer, then
$\{w_1\}$ lies in the domain $|w|<1$. Then the following
Proposition holds true
\begin{prop}\label{Proposition2.12}
The function $\underline{S}_{z,z',\xi,\theta=2}(x,y) $ in the
definition of the matrix kernel
$\underline{\mathbb{K}}_{z,z',\xi,\theta=2}(x,y)$ is antisymmetric
with respect to $x\longleftrightarrow y$, and it can be written as
\begin{equation}\label{Equation2.121}
\underline{S}_{z,z',\xi,\theta=2}(x,y)=\widetilde{\underline{S}}_{z,z',\xi,\theta=2}(x,y)
-\widetilde{\underline{S}}_{z,z',\xi,\theta=2}(y,x),
\end{equation}
where the function
$\widetilde{\underline{S}}_{z,z',\xi,\theta=2}(x,y)$ is defined by
equations (\ref{Equation2.3.1.1})-(\ref{Equation2.3.1.5}).
\end{prop}
\subsubsection{} Let us drop assumption that the parameters $z, z'$ are
admissible. Denote by
$\varrho_n^{(z,z',\xi,\theta=2)}(x_1,\ldots,x_n)$ the function
which is obtained from
$\varrho_n^{(z,z',\xi,\theta=2)}(x_1,\ldots,x_n)$ by setting
$z'=z-1$ in the formulae of Theorem \ref{MAINTHEOREM1}.
\begin{prop}\label{Proposition2.13}
The function $\varrho_n^{(z,z',\xi,\theta=2)}(x_1,\ldots,x_n)$
takes the form
$$
\rho_n^{(z,z'=z-1,\xi,\theta=2)}(x_1,\ldots,x_n)=\Pf\left[\widehat{\underline{\mathbb{K}}}_{z,z'=z-1,\xi,\theta=2}(x_i,x_j)\right]_{i,j=1}^n,
$$
where
\begin{equation}
\widehat{\underline{\mathbb{K}}}_{z,z'=z-1,\xi,\theta=2}(x,y)=\left[\begin{array}{cc}
  \widehat{\underline{S}}_{z,z'=z-1,\xi,\theta=2}(x,y) & -\widehat{\underline{S}}_{z,z'=z-1,\xi,\theta=2}(x,y+1) \\
   -\widehat{\underline{S}}_{z,z'=z-1,\xi,\theta=2}(x+1,y) &  \widehat{\underline{S}}_{z,z'=z-1,\xi,\theta=2}(x+1,y+1)\\
\end{array}\right],
\end{equation}
and the function
$\widehat{\underline{S}}_{z,z'=z-1,\xi,\theta=2}(x,y)$ can be
written as
\begin{equation}\label{ContorS4z'=z-1}
\begin{split}
\widehat{\underline{S}}_{z,z'=z-1,\xi,\theta=2}(x,y)=-\frac{1}{(2\pi
i)^2}\oint\limits_{\{w_1\}}\oint\limits_{\{w_2\}}\frac{(1-\sqrt{\xi}w_1)^{-z'}(1-\frac{\sqrt{\xi}}{w_1})^z
(1-\sqrt{\xi}w_2)^{-z}(1-\frac{\sqrt{\xi}}{w_2})^{z'}}{w_1w_2-1}\\
\times\frac{(w_2-w_1)}{(w_2^2-1)(w_1^2-1)}\frac{dw_1}{w_1^{x-\frac{1}{2}}}\frac{dw_2}{w_2^{y-\frac{1}{2}}}.
\end{split}
\end{equation}
 where
  $\{w_1\}$ and $\{w_2\}$ are arbitrary simple
contours satisfying the following conditions
\begin{itemize}
    \item  both contours go around $0$ in positive direction;
    \item  the point $\xi^{1/2}$ is in the interior of each of the
    contours while the point $\xi^{-1/2}$ lies outside the
    contours;
    \item Both contours  $\{w_1\}$ and $\{w_2\}$ lie in the domain $|w|>1$.
\end{itemize}
\end{prop}
Formula (\ref{ContorS4z'=z-1})  is equivalent to the result of
Theorem 3.1 a) in Strahov \cite{strahov}.  Theorem 3.1 a) in
Strahov \cite{strahov} was obtained by a completely different
method, and this comparison provides a check of validity for
Theorem \ref{MAINTHEOREM1}.
\subsection{The method: analytic
continuation of the Meixner symplectic ensemble} It was shown in
Borodin and Strahov \cite{BS} that the $z$-measures with
parameters $z=2N$, $z'=2N+\beta-2$, and $\theta=2$ turns into an
ensemble of $N$ particles on $\Zp$ called in  Borodin and Strahov
\cite{BS} the Meixner symplectic ensemble. It was shown in
\cite{BS} that this discrete ensemble is integrable in the sense
that the correlation function can be expressed explicitly in terms
of known functions. Namely, a discrete version of the method
developed by Tracy and Widom \cite{tracy}, Widom \cite{widom}
works for the Meixner symplectic ensemble, and correlation
functions are expressible in terms of Pfaffians of $2\times 2$
matrix kernels. The matrix elements of these kernels can be
written in terms of the classical Meixner orthogonal polynomials.
In the present paper we provide  contour integral representations
for the elements of the correlation kernel (see Theorem
\ref{TheoremSN4}), which is the result of an independent interest.

We regard the $z$-measures  with the Jack parameter $\theta=2$ as
the result of analytic continuation of the Mexiner symplectic
ensemble in parameter $N$ (number of particles). The procedure of
the analytic continuation is a natural extension of the approach
developed in Borodin and Olshanski \cite{BO} to much more
complicated situation of the matrix correlation kernels.

\section{The relation
between the $z$-measure with the parameter $\theta=2$ and the
Meixner symplectic ensemble. The correlation function for the
Meixner symplectic ensemble} We define the Meixner symplectic
ensemble in the same way as in Borodin and Strahov \cite{BS},
Section 2. Elements of $\Zp=\{0,1,2,\ldots\}$ will be denoted by
letters $\tx$, $\ty$. (Recall that the elements of $\Z'$ were
denoted by letters $x, y$.)

Let $w(\tx)$ be a strictly positive real valued function defined
on $\Zp$ with finite moments, i.e. the series
$\sum_{\tx\in\Zp}w(\tx)\tx^{j}$ converges for all $j=0,1,\ldots$.
\begin{defn}\label{1DEF}
The $N$-point discrete symplectic ensemble with the weight
function $w$ and the phase space $\Zp$ is the random $N$-point
configuration in $\Zp$ such that the probability of a particular
configuration $\tx_1<\ldots <\tx_N$ is given by
$$
\Prob\left\{\tx_1,\ldots,\tx_N\right\}=Z_{N4}^{-1}\;\prod\limits_{i=1}^Nw(\tx_i)
\prod\limits_{1\leq i<j\leq
N}(\tx_i-\tx_j)^2(\tx_i-\tx_j-1)(\tx_i-\tx_j+1).
$$
Here $Z_{N4}$ is a normalization constant which is assumed to be
finite.
\end{defn}
In what follows $Z_{N4}$ is referred to as the partition function
of the discrete symplectic ensemble under considerations.

We consider the particular case when $w(\tx)$ is the Meixner
weight given by
 the formula
\begin{equation}\label{MeinerWeight171}
W^{\Meixner}_{\beta,\xi}(\tx)=\frac{(\beta)_{\tx}}{\tx!}\xi^{\tx},\;\;\tx\in\Zp,
\end{equation}
where $\beta$ is a strictly positive real parameter, and
$0<\xi<1$. In this situation we say that we are dealing with the
Meixner symplectic ensemble.
\begin{prop}\label{PROPOSITIONI}
For $N=1,2,\ldots $ let $\Y(N)\subset\Y$ denote the set of
diagrams $\lambda$ with  $l(\lambda)\leq N$ (where $l(\lambda)$ is
the number of rows in $\lambda$). Under the bijection between
diagrams $\lambda\in\Y(N)$ and $N$-point configurations on $\Zp$
defined by
$$
\lambda\longleftrightarrow \tilde{x}_{N-i+1}=\lambda_i-2i+2N\;\;
(i=1,\ldots, N)
$$
the $z$-measure with parameters $z=2N$, $\theta=2$,
$z'=2N+\beta-2$ turns into
$$
\Prob^{\Meixner}\left\{\tilde{x}_1,\ldots,\tilde{x}_N\right\}=\const\cdot\prod\limits_{i=1}^N\frac{(\beta)_{\tilde{x}_i}}{\tilde{x}_i!}\xi^{\tilde{x}_i}
\prod\limits_{1\leq i\leq j\leq
N}(\tilde{x}_i-\tilde{x}_j)^2(\tilde{x}_i-\tilde{x}_j-1)(\tilde{x}_i-\tilde{x}_j+1),
$$
which is precisely the discrete symplectic ensemble with the
Meixner weight in the sense of Definition \ref{1DEF}.
\end{prop}
\begin{proof}
The proof is  a  straightforward computation based on the
application of the explicit formulae for
$H(\lambda;2)H'(\lambda;2)$, see the proof of Lemma 3.5 in
\cite{BO}, and $(z)_{\lambda,\theta}$, see Section 1 in \cite{BO}.
\end{proof}
We employ the same notation for the Meixner polynomials as in
Borodin and Olshanski \cite{BO}. Thus the Meixner polynomials are
denoted by $\mathfrak{M}_n(\tx;\beta,\xi)$. We use the same
normalization for these polynomials as in Koekoek and Swarttouw
\cite{koekoek}). Note that in Koekoek and Swarttouw
\cite{koekoek}) the parameter $\xi$ in the definition of the
Meixner weight is denoted as $c$. For basic properties of the
classical discrete orthogonal polynomials, and, in particular, the
Meixner polynomials, see Ismail \cite{ismail}.

As in Borodin and Olshanski \cite{BO}, we set
$$
\tMf_n(\tx;\beta,\xi)=(-1)^n\frac{\tMf_n(\tx;\beta,\xi)}{||\tMf_n(.\;;\beta,\xi)||}
\sqrt{W^{\Meixner}_{\beta,\xi}(\tx)},\; \tx\in\Zp,
$$
where
$$
||\tMf_n(.\;;\beta,\xi)||^2=\sum\limits_{\tx=0}^{\infty}\Mf_n^2(\tx;\beta,\xi)W^{\Meixner}_{\beta,\xi}(\tx).
$$
Let $\HC_{\Meixner}$ be the space spanned by functions $\tMf_0,
\tMf_1, \tMf_2,\ldots$, that is, each element of $\HC_{\Meixner}$
is a linear combination of $\tMf_0, \tMf_1, \tMf_2,\ldots$. We
introduce the operators $D_+^{\Meixner}$, $D_-^{\Meixner}$, and
$E^{\Meixner}$ which act on the elements of the space
$\HC_{\Meixner}$. The first and the second operators,
$D_+^{\Meixner}$ and $D_-^{\Meixner}$, are defined by the
expression:
$$
\left(D_{\pm}^{\Meixner}f\right)(\tx)=\sum\limits_{\ty=0}^{\infty}D_{\pm}^{\Meixner}(\tx,\ty)f(\ty),
$$
where the kernels $D_{\pm}^{\Meixner}(\tx,\ty)$ are given
explicitly by
\begin{equation}\label{DPLUS(x,y)}
D_+^{\Meixner}(\tx,\ty)=\frac{1}{\sqrt{\xi}}\sqrt{\frac{1+\tx}{\beta+\tx}}\;\delta_{\tx+1,\ty},\;\;\tx,
\ty\in\Zp,
\end{equation}
\begin{equation}\label{DMINUS(x,y)}
D_-^{\Meixner}(\tx,\ty)=\frac{1}{\sqrt{\xi}}\sqrt{\frac{\tx}{\beta+\tx-1}}\;\delta_{\tx-1,\ty},\;\;\tx,
\ty\in\Zp.
\end{equation}
The third operator, $E^{\Meixner}$, is defined by the formula
\begin{equation}\label{EMeixner}
\begin{split}
\left(E^{\Meixner} f\right)(\tx)=
\left\{%
\begin{array}{ll}
    -\sqrt{\xi}\sum\limits_{\ty=0}^{\infty}
\sqrt{\frac{(\beta+\tx)_{\ty+1,2}(2+\tx)_{\ty,2}}{(1+\tx)_{\ty+1,2}(1+\beta+\tx)_{\ty,2}}}\,f(\tx+2\ty+1), & \tx\;\; \hbox{is even;}\\
     \\\sqrt{\xi}\sum\limits_{\ty=0}^{\infty}
\sqrt{\frac{(1-\beta-\tx)_{\ty+1,2}(1-\tx)_{\ty,2}}{(-\tx)_{\ty+1,2}(2-\beta-\tx)_{\ty,2}}}\,f(\tx-2\ty-1),
& \tx\;\; \hbox{is odd.}
\end{array}%
\right.
\end{split}
\end{equation}
Note that the sum in the case of an odd $\tx$ actually runs from
$0$ to $\frac{\tx-1}{2}$, so $\tx-2\ty-1\in\Zp$ in the argument of
the function $f$. It is explained in Borodin and Strahov \cite{BS}
that the series defining $(E^{\Meixner})f(x)$ converges for any
$f$ from $\HC_{\Meixner}$, i.e. $E^{\Meixner} f$ is well defined,
see the discussion after equation (2.3) in Borodin and Strahov
\cite{BS}, Section 2.

Let $\HC^{\Meixner}_{2N}$ be the subspace of $\HC_{\Meixner}$
spanned by the functions $\tMf_0$, $\tMf_1,\ldots ,\tMf_{2N-1}$.
We denote by $K_{2N}^{\Meixner}$ the projection operator onto
$\HC_{2N}^{\Meixner}$. Its kernel is
$$
K_{2N}^{\Meixner}(\tx,\ty)=\sum\limits_{k=0}^{2N-1}\tMf_k(\tx)\tMf_k(\ty).
$$

In addition, we introduce the operator $S_{2N}^{\Meixner}$ by the
formula
\begin{equation}\label{SMeixner}
\begin{split}
S_{2N}^{\Meixner}=E^{\Meixner}
K_{2N}^{\Meixner}&+K_{2N}^{\Meixner}E^{\Meixner}\\
&-E^{\Meixner}
K_{2N}^{\Meixner}D^{\Meixner}K_{2N}^{\Meixner}E^{\Meixner},
\end{split}
\end{equation}
where
\begin{equation}\label{DMeixner}
D^{\Meixner}=D_+^{\Meixner}-D_-^{\Meixner}.
\end{equation}
It can be checked that the operators $E^{\Meixner}$ and
$D^{\Meixner}$ are mutually inverse.  The operator
$S_{2N}^{\Meixner}$ acts in the same space $\HC_{\Meixner}$.

Let $\Conf(\Zp)$ denote a collection of sets each of which is
itself a collection of $N$ pairwise distinct points from $\Zp$.
The Meixner symplectic ensemble is a probability measure on
$\Conf(\Zp)$, and its  $n$th correlation function,
$\varrho_{n,\Meixner}^{(N,\beta,\xi)}$, is defined by
$$
\varrho_{n,\Meixner}^{(N,\beta,\xi)}(\tx_1,\ldots
,\tx_n)=\Prob_{\Meixner}\left\{\mbox{the random configuration
contains}\; \tx_1,\ldots, \tx_n\right\},
$$
where $\tx_1$, $\tx_2$, $\ldots$, $\tx_n$ are pairwise distinct
points of $\Zp$.
\begin{prop}\label{TheoremMeixnerKN4}
The correlation function of the $N$-point  Meixner symplectic
ensemble can be written as a Pfaffian of $2\times 2$ matrix valued
kernel,
$$
\varrho_{n,\Meixner}^{(N,\beta,\xi)}(\tx_1,\ldots
,\tx_n)=\Pf\left[\mathbb{K}_{2N}^{\Meixner}(\tx_i,\tx_j)\right]_{i,j=1}^n.
$$
This kernel, $\mathbb{K}_{2N}^{\Meixner}(\tx,\ty)$, has the
following representation
$$
\mathbb{K}_{2N}^{\Meixner}(\tx,\ty)=\left[\begin{array}{cc}
  S_{2N}^{\Meixner}(\tx,\ty) & -S_{2N}^{\Meixner}D_-^{\Meixner}(\tx,\ty) \\
  -D_+^{\Meixner}S_{2N}^{\Meixner}(\tx,\ty) & D_+^{\Meixner}S_{2N}^{\Meixner}D_-^{\Meixner}(\tx,\ty) \\
\end{array}\right],
$$
where the matrix entries are the kernels of the operators
$S_{2N}^{\Meixner}$, $-S_{2N}^{\Meixner}D_-^{\Meixner}$,
$-D_+^{\Meixner}S_{2N}^{\Meixner}$, and
$D_+^{\Meixner}S_{2N}^{\Meixner}D_-^{\Meixner}$.
\end{prop}
\begin{proof} The representation for $\mathbb{K}_{2N}^{\Meixner}$ follows from Theorem 2.4,
Theorem 2.9 and Proposition 14.3 in Borodin and Strahov \cite{BS}.
\end{proof}
\begin{lem}\label{LemmaRelation}
For $N=1, 2,\ldots $ let $z=2N$ and $z'=2N+\beta-2$ with
$\beta>0$. Assume that $x_1,\ldots ,x_n$ lie in the subset
$\Zp-2N+1/2\subset\Z'$, so that the points $\tx_i=x_i+2N-1/2$ are
in $\Zp$. Then
$$
\varrho_{n}^{(z,z',\xi,\theta=2)}(x_1,\ldots
,x_n)=\Pf\left[\mathbb{K}_{2N}^{\Meixner}\left(x_i+2N-\frac{1}{2},
x_j+2N-\frac{1}{2}\right)\right]_{i,j=1}^n.
$$
\end{lem}
\begin{proof}If $z=2N$ and $z'=2N+\beta-2$, then Proposition \ref{PROPOSITIONI} implies that
$M_{z,z',\xi,\theta=2}$ defines the Meixner symplectic ensemble on
the point configurations $\tilde{X}(\lambda)$ defined by
$$
\tx_{N-i+1}=\lambda_i-2i+2N,\;\; i=1,\ldots, N.
$$
To obtain $\varrho_{n}^{z,z',\xi,\theta=2}(x_1,\ldots ,x_n)$ we
need to consider random configurations $\DC_2(\lambda)$ defined by
$x_i=\lambda_i-2i+1/2$, where $i=1,\ldots, N$. On the other hand,
there is a bijective correspondence between the set of all
configurations $\tilde{X}(\lambda)$, and the set of all
configurations $\DC_2(\lambda)$ defined by
\begin{equation}\label{correspondence}
\tx_{N-i+1}=x_i-\frac{1}{2}+2N,\;\; i=1,\ldots, N.
\end{equation}
Note that two configurations $\tilde{X}(\lambda)$ and
$\DC_2(\lambda)$ related by (\ref{correspondence}) have the same
probability. The statement of the Lemma   immediately follows from
this observation, and from Proposition \ref{TheoremMeixnerKN4} .
\end{proof}
\section{The contour integral representation for $S_{2N}^{\Meixner}(\tx,\ty)$}

The aim of this Section is to obtain an explicit formula for the
function $S_{2N}^{\Meixner}(\tx,\ty)$ which completely determines
the correlation function for the Meixner symplectic ensemble via
Proposition \ref{TheoremMeixnerKN4}. Namely, we provide a contour
integral representation for $S_{2N}^{\Meixner}(\tx,\ty)$.
\begin{thm}\label{TheoremSN4}
The function $S_{2N}^{\Meixner}(\tx,\ty)$ (which is the kernel of
the operator $S_{2N}^{\Meixner}$) admits the following contour
integral representation \\
a) If both $\tx$ and $\ty$ are even, then
\begin{equation}
\begin{split}
&S_{2N}^{\Meixner}(\tx,\ty)=-\frac{\sqrt{\xi}}{(2\pi
i)^2}\sqrt{\frac{\Gamma(\tx+1)\Gamma(\ty+\beta)}{\Gamma(\tx+\beta)\Gamma(\ty+1)}}\\
&\biggl[\oint\limits_{\{w_1\}}\oint\limits_{\{w_2\}}\frac{(1-\sqrt{\xi}w_1)^{-2N-\beta+1}(1-\frac{\sqrt{\xi}}{w_1})^{2N}
(1-\sqrt{\xi}w_2)^{-2N}(1-\frac{\sqrt{\xi}}{w_2})^{2N+\beta-1}}{w_1w_2-1}\\
&F\left(\frac{\tx+2}{2},1;\frac{\tx+\beta+1}{2};\frac{1}{w_1^2}\right)
\frac{dw_1}{w_1^{\tx-2N+2}}\frac{dw_2}{w_2^{\ty-2N+1}}\\
&-\oint\limits_{\{w_1\}}\oint\limits_{\{w_2\}}(1-\sqrt{\xi}w_1)^{-2N-\beta+1}(1-\frac{\sqrt{\xi}}{w_1})^{2N}
(1-\sqrt{\xi}w_2)^{-2N}(1-\frac{\sqrt{\xi}}{w_2})^{2N+\beta-1}\\
&F\left(\frac{\tx+2}{2},1;\frac{\tx+\beta+1}{2};\frac{1}{w_1^2}\right)
\left[F\left(\frac{\ty+\beta}{2},1;\frac{\ty+1}{2};\frac{1}{w_2^2}\right)-1\right]
\frac{dw_1}{w_1^{\tx-2N+2}}\frac{dw_2}{w_2^{\ty-2N+1}}\biggr].
\end{split}
\nonumber
\end{equation}
In the formula just written  $\{w_1\}$ and $\{w_2\}$ are arbitrary
simple contours satisfying the following conditions
\begin{itemize}
    \item  both contours go around $0$ in positive direction;
    \item  the point $\xi^{1/2}$ is in the interior of each of the
    contours while the point $\xi^{-1/2}$ lies outside the
    contours;
    \item the contour $\{w_1^{-1}\}$ is contained in the interior
    of the contour $\{w_2\}$ (equivalently, $\{w_2^{-1}\}$ is
    contained in the interior of $\{w_1\}$);
    \item both contours $\{w_1\}, \{w_2\}$ lie in the domain $|w|>1$.
\end{itemize}
b) In the case when both $\tx$ and $\ty$ are odd positive integers
we have
\begin{equation}
\begin{split}
&S_{2N}^{\Meixner}(\tx,\ty)=\frac{\sqrt{\xi}}{(2\pi
i)^2}\sqrt{\frac{\Gamma(\tx+1)\Gamma(\ty+\beta)}{\Gamma(\tx+\beta)\Gamma(\ty+1)}}\\
&\biggl[\oint\limits_{\{w_1\}}\oint\limits_{\{w_2\}}\frac{(1-\sqrt{\xi}w_1)^{-2N-\beta+1}(1-\frac{\sqrt{\xi}}{w_1})^{2N}
(1-\sqrt{\xi}w_2)^{-2N}(1-\frac{\sqrt{\xi}}{w_2})^{2N+\beta-1}}{w_1w_2-1}\\
&\left[F\left(-\frac{\beta+\tx-1}{2},
1;-\frac{\tx}{2};\frac{1}{w_1^2}\right)-1\right]
\frac{dw_1}{w_1^{\tx-2N+2}}\frac{dw_2}{w_2^{\ty-2N+1}}\\
&+\oint\limits_{\{w_1\}}\oint\limits_{\{w_2\}}(1-\sqrt{\xi}w_1)^{-2N-\beta+1}(1-\frac{\sqrt{\xi}}{w_1})^{2N}
(1-\sqrt{\xi}w_2)^{-2N}(1-\frac{\sqrt{\xi}}{w_2})^{2N+\beta-1}\\
& \left[F\left(-\frac{\beta+\tx-1}{2},
1;-\frac{\tx}{2};\frac{1}{w_1^2}\right)-1\right]F\left(\frac{1-\ty}{2},1;\frac{2-\beta-\ty}{2};w_2^2\right)
\frac{dw_1}{w_1^{\tx-2N+2}}\frac{dw_2}{w_2^{\ty-2N+21}}\biggr],
\end{split}
\nonumber
\end{equation}
where the contours $\{w_1\}$, $\{w_2\}$  are arbitrary simple
contours satisfying the first three conditions of a) that lie in
the domain $|w|<1$.

c) If $\tx$ is even positive integer, and $\ty$ is an  odd
positive integer, then
\begin{equation}
\begin{split}
&S_{2N}^{\Meixner}(\tx,\ty)=-\frac{\sqrt{\xi}}{(2\pi
i)^2}\sqrt{\frac{\Gamma(\tx+1)\Gamma(\ty+\beta)}{\Gamma(\tx+\beta)\Gamma(\ty+1)}}\\
&\biggl[\oint\limits_{\{w_1\}}\oint\limits_{\{w_2\}}\frac{(1-\sqrt{\xi}w_1)^{-2N-\beta+1}(1-\frac{\sqrt{\xi}}{w_1})^{2N}
(1-\sqrt{\xi}w_2)^{-2N}(1-\frac{\sqrt{\xi}}{w_2})^{2N+\beta-1}}{w_1w_2-1}\\
&F\left(\frac{\tx+2}{2},1;\frac{\tx+\beta+1}{2};\frac{1}{w_1^2}\right)
\frac{dw_1}{w_1^{\tx-2N+2}}\frac{dw_2}{w_2^{\ty-2N+1}}\\
&+\oint\limits_{\{w_1\}}\oint\limits_{\{w_2\}}(1-\sqrt{\xi}w_1)^{-2N-\beta+1}(1-\frac{\sqrt{\xi}}{w_1})^{2N}
(1-\sqrt{\xi}w_2)^{-2N}(1-\frac{\sqrt{\xi}}{w_2})^{2N+\beta-1}\\
&F\left(\frac{\tx+2}{2},1;\frac{\tx+\beta+1}{2};\frac{1}{w_1^2}\right)
F\left(\frac{1-\ty}{2},1;\frac{2-\beta-\ty}{2}; w_2^2\right)
\frac{dw_1}{w_1^{\tx-2N+2}}\frac{dw_2}{w_2^{\ty-2N+1}}\biggr],
\end{split}
\nonumber
\end{equation}
where the contours $\{w_1\}$, $\{w_2\}$  are arbitrary simple
contours satisfying the first three conditions of a). Moreover,
the first contour, $\{w_1\}$, lies in the domain  $|w|>1$, and the
second contour, $\{w_2\}$, lies in the domain $|w|<1$.

d) Finally, if $\tx$ is an odd integer, and $\ty$ is even integer,
then
\begin{equation}
\begin{split}
&S_{2N}^{\Meixner}(\tx,\ty)=\frac{\sqrt{\xi}}{(2\pi
i)^2}\sqrt{\frac{\Gamma(\tx+1)\Gamma(\ty+\beta)}{\Gamma(\tx+\beta)\Gamma(\ty+1)}}\\
&\biggl[\oint\limits_{\{w_1\}}\oint\limits_{\{w_2\}}\frac{(1-\sqrt{\xi}w_1)^{-2N-\beta+1}(1-\frac{\sqrt{\xi}}{w_1})^{2N}
(1-\sqrt{\xi}w_2)^{-2N}(1-\frac{\sqrt{\xi}}{w_2})^{2N+\beta-1}}{w_1w_2-1}\\
&\left[F\left(-\frac{\tx+\beta-1}{2},1;-\frac{\tx}{2};w_1^2\right)-1\right]
\frac{dw_1}{w_1^{\tx-2N+2}}\frac{dw_2}{w_2^{\ty-2N+1}}\\
&-\oint\limits_{\{w_1\}}\oint\limits_{\{w_2\}}(1-\sqrt{\xi}w_1)^{-2N-\beta+1}(1-\frac{\sqrt{\xi}}{w_1})^{2N}
(1-\sqrt{\xi}w_2)^{-2N}(1-\frac{\sqrt{\xi}}{w_2})^{2N+\beta-1}\\
&\left[F\left(-\frac{\tx+\beta-1}{2},1;-\frac{\tx}{2};w_1^2\right)-1\right]
\left[F\left(\frac{\beta+\ty}{2},1;\frac{\ty+1}{2};\frac{1}{w_2^2}\right)-1\right]
\frac{dw_1}{w_1^{\tx-2N+2}}\frac{dw_2}{w_2^{\ty-2N+1}}\biggr],
\end{split}
\nonumber
\end{equation}
where the contours $\{w_1\}$, $\{w_2\}$  are arbitrary simple
contours satisfying the first three conditions of a), which both
lie in the domain $|w|<1$.
\end{thm}
\begin{proof}
We start from equation (\ref{SMeixner}). The operators
$E^{\Meixner}, D^{\Meixner}$ are defined explicitly by equations
(\ref{DPLUS(x,y)}), (\ref{DMINUS(x,y)}), (\ref{EMeixner}), and
(\ref{DMeixner}). The contour integral representation for the
kernel $K_{2N}^{\Meixner}$ can be obtained immediately from
Proposition \ref{EquationRepresentationKzz'}. Indeed, Theorem 3.2
and Proposition 2.8 in Borodin and Olshanski \cite{BO} imply the
relation
\begin{equation}\label{Equation4.11}
K_{2N}^{\Meixner}(\tx,\ty)=\underline{K}_{z=2N,z'=2N+\beta-1,\xi}(\tx-2N+1,\ty-2N+1).
\end{equation}
Once the contour integral representation of
$K_{2N}^{\Meixner}(\tx,\ty)$ is given it is straightforward to
derive the contour integral representations for the kernels of the
operators $E^{\Meixner}K_{2N}^{\Meixner}$ and
$K_{2N}^{\Meixner}E^{\Meixner}$. Next we need to derive the
contour integral representation for the kernel of the operator
$D^{\Meixner}K_{2N}^{\Meixner}D^{\Meixner}$ (this is the most
nontrivial part of these calculations). We have
$$
K_{2N}^{\Meixner}D^{\Meixner}K_{2N}^{\Meixner}=K_{2N}^{\Meixner}D^{\Meixner}_+K_{2N}^{\Meixner}
-K_{2N}^{\Meixner}D^{\Meixner}_-K_{2N}^{\Meixner}.
$$
Let us first derive the contour integral representation for
$K_{2N}^{\Meixner}D_+K_{2N}^{\Meixner}(\tx,\ty)$. Taking into
account the definition of $D_{+}^{\Meixner}$ we can write
$$
K_{2N}^{\Meixner}D_+^{\Meixner}K_{2N}^{\Meixner}(\tx,\ty)
=\frac{1}{\sqrt{\xi}}\sum\limits_{m=0}^{+\infty}K_{2N}^{\Meixner}(\tx,m)\sqrt{\frac{m+1}{m+\beta}}K_{2N}^{\Meixner}(m+1,\ty).
$$
Using the contour integral representation for
$K_{2N}^{\Meixner}(\tx,m)$ and $K_{2N}^{\Meixner}(m+1,\ty)$ we
obtain
\begin{equation}
\begin{split}
&K_{2N}^{\Meixner}D_+^{\Meixner}K_{2N}^{\Meixner}(\tx,\ty)=
\frac{1}{(2\pi
i)^4\sqrt{\xi}}\sqrt{\frac{\Gamma(\tx+\beta)\Gamma(\ty+1)}{\Gamma(\tx+1)\Gamma(\ty+\beta)}}\\
&\times \oint\limits_{\{w_1\}}\oint\limits_{\{w_2\}}
\oint\limits_{\{w_3\}}\oint\limits_{\{w_4\}}
\Phi(w_1,w_2)\Phi(w_3,w_4)\left(\sum\limits_{m=0}^{+\infty}\frac{1}{(w_2w_3)^m}\right)
\frac{dw_1dw_2dw_3dw_4}{w_1^{\tx-2N+1}w_2^{-2N+1}w_3^{-2N+2}w_4^{\ty-2N+1}},
\end{split}
\nonumber
\end{equation}where
\begin{equation}\label{EquationPhi}
\Phi(w_1,w_2)=\frac{(1-\sqrt{\xi}w_1)^{-2N}(1-\frac{\sqrt{\xi}}{w_1})^{2N+\beta-1}
(1-\sqrt{\xi}w_2)^{-2N-\beta+1}(1-\frac{\sqrt{\xi}}{w_2})^{2N}}{w_1w_2-1}.
\end{equation}
We observe that the integral above remains unchanged if we replace
the sum in the integrand by
$$
\sum\limits_{m=-\infty}^{+\infty}\frac{1}{(w_2w_3)^m}.
$$
We split this sum into two parts,
$$
\sum\limits_{m=-\infty}^{+\infty}\frac{1}{(w_2w_3)^m}=\sum\limits_{m=0}^{+\infty}\frac{1}{(w_2w_3)^m}
+\sum\limits_{m=-\infty}^{-1}\frac{1}{(w_2w_3)^m}.
$$
If $|w_2w_3|>1$, then the first sum in the righthand side of the
equation converges, and it equals
$$
\sum\limits_{m=0}^{+\infty}\frac{1}{(w_2w_3)^m}=\frac{w_2}{w_2-w_3^{-1}}.
$$
If $|w_2w_3|<1$, then the second sum can be written as
$$
\sum\limits_{m=-\infty}^{-1}\frac{1}{(w_2w_3)^m}=\frac{w_2}{w_3^{-1}-w_2}.
$$
This gives us
\begin{equation}\label{EquationKND+KN}
\begin{split}
&K_{2N}^{\Meixner}D_+^{\Meixner}K_{2N}^{\Meixner}(\tx,\ty)=
\frac{1}{(2\pi
i)^4\sqrt{\xi}}\sqrt{\frac{\Gamma(\tx+\beta)\Gamma(\ty+1)}{\Gamma(\tx+1)\Gamma(\ty+\beta)}}\\
&\times \biggl[\underset{|w_2|>|w_3^{-1}|}{\oint\oint\oint\oint}
\Phi(w_1,w_2)\Phi(w_3,w_4)\frac{w_2}{w_2-w_3^{-1}}
\frac{dw_1dw_2dw_3dw_4}{w_1^{\tx-2N+1}w_2^{-2N+1}w_3^{-2N+2}w_4^{\ty-2N+1}}\\
&+\underset{|w_2|<|w_3^{-1}|}{\oint\oint\oint\oint}
\Phi(w_1,w_2)\Phi(w_3,w_4)\frac{w_2}{w_3^{-1}-w_2}
\frac{dw_1dw_2dw_3dw_4}{w_1^{\tx-2N+1}
w_2^{-2N+1}w_3^{-2N+2}w_4^{\ty-2N+1}}\biggr].
\end{split}
\end{equation}
Here we take as the contours concentric circles  $\{w_1\}$,
$\{w_2\}$, $\{w_3\}$, $\{w_4\}$. The contours  $\{w_1\}$ and
$\{w_2\}$ satisfy the same conditions as in the statement of
Proposition \ref{EquationRepresentationKzz'}, in particular we can
agree that $\frac{1}{|w_2|}<|w_1|$. We also agree that the
contours $\{w_3\}$ and $\{w_4\}$ satisfy the conditions of
Proposition \ref{EquationRepresentationKzz'}, and that
$\frac{1}{|w_3|}<|w_4|$. In addition, we require that
$|w_3|>\frac{1}{|w_2|}$ in first integral which corresponds to the
sum over $\Zp$.  In the second integral (corresponding to the sum
over $\Z_{<0}$) we chose contours in such a way that
$|w_3|<\frac{1}{|w_2|}$.

We transform the first integral: keeping the contours $\{w_1\}$,
$\{w_3\}$, $\{w_4\}$ unchanged we move $\{w_2\}$ inside the circle
of the radius $\frac{1}{|w_3|}$. Then we obtain an integral which
cancels the second integral in equation (\ref{EquationKND+KN}),
plus an integral arising from the residue of the function $
w_2\longrightarrow (w_2-w_3^{-1})^{-1}$.  This gives us
\begin{equation}
\begin{split}
&K_{2N}^{\Meixner}D_+^{\Meixner}K_{2N}^{\Meixner}(\tx,\ty)\\
&=\frac{1}{(2\pi
i)^3\sqrt{\xi}}\sqrt{\frac{\Gamma(\tx+\beta)\Gamma(\ty+1)}{\Gamma(\tx+1)\Gamma(\ty+\beta)}}
\oint\limits_{\{w_1\}}\oint\limits_{\{w_3\}}\oint\limits_{\{w_4\}}
\frac{\Phi(w_1,w_3^{-1})\Phi(w_3,w_4)}{w_3}
\frac{dw_1dw_3dw_4}{w_1^{\tx-2N+1}w_3w_4^{\ty-2N+1}}.
\end{split}
\nonumber
\end{equation}
We find
\begin{equation}
\begin{split}
\Phi(w_1,w_3^{-1})\Phi(w_3,w_4)=\frac{(1-\sqrt{\xi}w_1)^{-2N}(1-\frac{\sqrt{\xi}}{w_1})^{2N+\beta-1}
(1-\sqrt{\xi}w_4)^{-2N-\beta+1}(1-\frac{\sqrt{\xi}}{w_4})^{2N}}{(w_1w_3^{-1}-1)(w_3w_4-1)}.
\end{split}
\nonumber
\end{equation}
Now we integrate over $\{w_3\}$. Note that the contour $\{w_3\}$
can always be chosen inside the circle $\{w_1\}$. Therefore the
integration over $\{w_3\}$ reduces to the computation of the
residue of the function $ w_3\longrightarrow (w_1-w_3)^{-1} $ in
the situation when $\{w_3\}$ lies inside $\{w_1\}$. The result is
\begin{equation}
\begin{split}
&K_{2N}^{\Meixner}D_+^{\Meixner}K_{2N}^{\Meixner}(\tx,\ty)=\frac{1}{(2\pi
i)^2\sqrt{\xi}}\sqrt{\frac{\Gamma(\tx+\beta)\Gamma(\ty+1)}{\Gamma(\tx+1)\Gamma(\ty+\beta)}}\\
&\times\oint\limits_{\{w_1\}}\oint\limits_{\{w_2\}}
\frac{(1-\sqrt{\xi}w_1)^{-2N}(1-\frac{\sqrt{\xi}}{w_1})^{2N+\beta-1}
(1-\sqrt{\xi}w_4)^{-2N-\beta+1}(1-\frac{\sqrt{\xi}}{w_4})^{2N}}{w_1w_4-1}
\frac{dw_1}{w_1^{\tx-2N+2}}\frac{dw_4}{w_4^{\ty-2N+1}}.
\end{split}
\nonumber
\end{equation}
Note that as soon as $\{w_3\}$ is chosen to be inside $\{w_1\}$,
and $|w_3|>|w_4|^{-1}$, we have $|w_1|>|w_4|^{-1}$ in the integral
above. Thus $\{w_4^{-1}\}$ is contained in the interior of
$\{w_1\}$. To obtain formula for
$K_{2N}^{\Meixner}D_+^{\Meixner}K_{2N}^{\Meixner}(\tx,\ty)$ we can
use the relation
$$
K_{2N}^{\Meixner}D_+^{\Meixner}K_{2N}^{\Meixner}(\tx,\ty)=
K_{2N}^{\Meixner}D_-^{\Meixner}K_{2N}^{\Meixner}(\ty,\tx)
$$
which follows immediately from the definitions of the involved
operators.  Thus we arrive to the formula
\begin{equation}\label{KDK}
\begin{split}
&K_{2N}^{\Meixner}D^{\Meixner}K_{2N}^{\Meixner}(\tx,\ty)=\frac{1}{(2\pi
i)^2\sqrt{\xi}}\sqrt{\frac{\Gamma(\tx+\beta)\Gamma(\ty+1)}{\Gamma(\tx+1)\Gamma(\ty+\beta)}}\\
&\times\oint\limits_{\{w_1\}}\oint\limits_{\{w_2\}}
\frac{(1-\sqrt{\xi}w_1)^{-2N}(1-\frac{\sqrt{\xi}}{w_1})^{2N+\beta-1}
(1-\sqrt{\xi}w_2)^{-2N-\beta+1}(1-\frac{\sqrt{\xi}}{w_2})^{2N}}{w_1w_2-1}
\frac{dw_1}{w_1^{\tx-2N+2}}\frac{dw_2}{w_2^{\ty-2N+1}}\\
&-\frac{1}{(2\pi
i)^2\sqrt{\xi}}\sqrt{\frac{\Gamma(\tx+1)\Gamma(\ty+\beta)}{\Gamma(\tx+\beta)\Gamma(\ty+1)}}\\
&\times\oint\limits_{\{w_1\}}\oint\limits_{\{w_2\}}
\frac{(1-\sqrt{\xi}w_2)^{-2N}(1-\frac{\sqrt{\xi}}{w_2})^{2N+\beta-1}
(1-\sqrt{\xi}w_1)^{-2N-\beta+1}(1-\frac{\sqrt{\xi}}{w_1})^{2N}}{w_1w_2-1}
\frac{dw_1}{w_1^{\tx-2N+1}}\frac{dw_2}{w_2^{\ty-2N+2}},
\end{split}
\end{equation}
where the contours $\{w_1\}$, $\{w_2\}$ are chosen in the same way
as in the statement of Proposition
\ref{EquationRepresentationKzz'}. Exploiting the fact that
$K_{2N}^{\Meixner}(\tx,\ty)=K_{2N}^{\Meixner}(\ty,\tx)$ we can
rewrite this as a double integral. Applying the operators
$E^{\Meixner}$ we obtain the contour integral representation for
the kernel of the operator
$E^{\Meixner}K_{2N}^{\Meixner}D^{\Meixner}K_{2N}^{\Meixner}E^{\Meixner}$.
Adding expressions for the kernels of
$E^{\Meixner}K_{2N}^{\Meixner}$ and
$K_{2N}^{\Meixner}E^{\Meixner}$ we arrive to formulae in the
statement of the Theorem.
\end{proof}
It is possible to represent the function
$S_{2N}^{\Meixner}(\tx,\ty)$ in the form which is manifestly
antisymmetric with respect to $\tx\longleftrightarrow\ty$. With
this purpose in mind we introduce functions
$P^{\Meixner}(\tx,w,\beta)$ and $Q^{\Meixner}(\tx,w,\beta)$ as
follows. If $\tx$ is an even integer, then the functions
$P^{\Meixner}(\tx,w,\beta)$ and $Q^{\Meixner}(\tx,w,\beta)$ are
defined by the formulae
\begin{equation}
P^{\Meixner}(\tx,w,\beta)=\sqrt{\frac{\Gamma(\tx+1)}{\Gamma(\tx+\beta)}}
F\left(\frac{\tx+2}{2},1;\frac{\tx+\beta+1}{2};\frac{1}{w^2}\right)\frac{1}{w^{\tx+1}},
\end{equation}
and
\begin{equation}
Q^{\Meixner}(\tx,w,\beta)=\sqrt{\frac{\Gamma(\tx+\beta)}{\Gamma(\tx+1)}}\left(
F\left(\frac{\tx+\beta}{2},1;\frac{\tx+1}{2};\frac{1}{w^2}\right)-1\right)\frac{1}{w^{\tx}},
\end{equation}
If $\tx$ is an odd integer, then $P^{\Meixner}(\tx,w,\beta)$ and
$Q^{\Meixner}(\tx,w,\beta)$ are defined by
\begin{equation}
P^{\Meixner}(\tx,w,\beta)=-\sqrt{\frac{\Gamma(\tx+1)}{\Gamma(\tx+\beta)}}
\left(F\left(-\frac{\beta+\tx-1}{2},1;-\frac{\tx}{2};w^2\right)-1\right)\frac{1}{w^{\tx+1}},
\end{equation}
and
\begin{equation}
Q^{\Meixner}(\tx,w,\beta)=-\sqrt{\frac{\Gamma(\tx+\beta)}{\Gamma(\tx+1)}}
F\left(-\frac{\tx-1}{2},1;-\frac{\beta+\tx-2}{2};w^2\right)\frac{1}{w^{\tx}}.
\end{equation}
In addition, set
\begin{equation}\label{Equation4.8}
\begin{split}
\widetilde{S}_{2N}^{\Meixner}(\tx,\ty)=\frac{\sqrt{\xi}}{(2\pi
i)^2}
\oint\limits_{\{w_1\}}\oint\limits_{\{w_2\}}&\frac{(1-\sqrt{\xi}w_1)^{-2N-\beta+1}(1-\frac{\sqrt{\xi}}{w_1})^{2N}
(1-\sqrt{\xi}w_2)^{-2N}(1-\frac{\sqrt{\xi}}{w_2})^{2N+\beta-1}}{w_1w_2-1}\\
&\times
P^{\Meixner}(\tx,w,\beta)Q^{\Meixner}(\tx,w,\beta)\frac{dw_1}{w_1^{-2N+1}}\frac{dw_2}{w_2^{-2N+1}}.
\end{split}
\end{equation}
where the contours $\{w_1\},\{w_2\}$ are chosen as in the
statement of Proposition \ref{EquationRepresentationKzz'} with the
following additional conditions. If $\tx$ is an even integer, then
$\{w_1\}$ lies in the domain $|w|>1$. If $\tx$ is an odd integer,
then $\{w_1\}$ lies in the domain $|w|<1$. The  same condition is
imposed on $\{w_2\}$: if $\tx$ is an even integer, then $\{w_2\}$
lies in the domain $|w|>1$, and  if $\tx$ is an odd integer, then
$\{w_1\}$ lies in the domain $|w|<1$.
\begin{prop}\label{Proposition4.2}
The function $S_{2N}^{\Meixner}(\tx,\ty)$ can be
written as
$$
S_{2N}^{\Meixner}(\tx,\ty)=S_{2N}^{\Meixner}(\tx,\ty)-S_{2N}^{\Meixner}(\ty,\tx),
$$
where the function $S_{2N}^{\Meixner}(\tx,\ty)$ is defined by
equation (\ref{Equation4.8}).
\end{prop}
\begin{proof}
Using the fact that the operators $D^{\Meixner}$ and
$E^{\Meixner}$ are mutually inverse we obtain from equation
(\label{SMeixner}) the relation
$$
D^{\Meixner}S_{2N}^{\Meixner}D^{\Meixner}=D^{\Meixner}K_{2N}^{\Meixner}+
K_{2N}^{\Meixner}D^{\Meixner}-K_{2N}^{\Meixner}D^{\Meixner}K_{2N}^{\Meixner}.
$$
This relation (together with formulae (\ref{Equation4.11}),
(\ref{KDK}), and Proposition \ref{EquationRepresentationKzz'})
enables us to find an explicit formula for the kernel of the
operator $D^{\Meixner}S_{2N}^{\Meixner}D^{\Meixner}$. Namely,
\begin{equation}
\begin{split}
&D^{\Meixner}S_{2N}^{\Meixner}D^{\Meixner}(\tx,\ty)=\frac{1}{(2\pi
i)^2\sqrt{\xi}}\sqrt{\frac{\Gamma(\tx+\beta)}{\Gamma(\tx+1)}\frac{\Gamma(\ty+1)}{\Gamma(\ty+\beta)}}\\
&\times\oint\limits_{\{w_1\}}\oint\limits_{\{w_2\}}\frac{(1-\sqrt{\xi}w_1)^{-2N-\beta+1}(1-\frac{\sqrt{\xi}}{w_1})^{2N}
(1-\sqrt{\xi}w_2)^{-2N}(1-\frac{\sqrt{\xi}}{w_2})^{2N+\beta-1}}{w_1w_2-1}\frac{dw_1}{w_1^{-2N+1}}\frac{dw_2}{w_2^{-2N+1}}\\
&-(\tx\longleftrightarrow\ty).
\end{split}
\nonumber
\end{equation}
Applying $E^{\Meixner}$ to the both sides of the formula above we
obtain the representation for $S_{2N}^{\Meixner}(\tx,\ty)$ in the
manifestly antisymmetric form.
\end{proof}
\begin{rem} For our purposes, in particular, for the analytic
continuation of the Meixner symplectic ensemble the expression for
$S_{2N}^{\Meixner}(\tx,\ty)$ given in Theorem \ref{TheoremSN4} is
more convenient.
\end{rem}

\section{Proof of Theorem \ref{MAINTHEOREM1} for special values of
parameters $z$ and $z'$}\label{SECTION5} The aim of the this
Section is to show that for $N=1,2,\ldots$, $z=2N$, and
$z'=2N+\beta-2$ with $\beta>0$ the formula for the correlation
function $\varrho_{n}^{(z,z',\xi,\theta=2)}(x_1,\ldots ,x_n)$
obtained in Lemma \ref{LemmaRelation} is equivalent to the formula
for the correlation function
$\varrho_{n}^{(z,z',\xi,\theta=2)}(x_1,\ldots ,x_n)$ stated in
Theorem \ref{MAINTHEOREM1}. Once we show this equivalence, we
prove Theorem \ref{MAINTHEOREM1} for special values of $z$ and
$z'$. The transformation from the formula in Lemma
\ref{LemmaRelation} (where the kernel $\mathbb{K}_{2N}^{\Meixner}$
is given by Proposition \ref{TheoremMeixnerKN4} together with
formula (\ref{SMeixner})) to the formula in Theorem
\ref{MAINTHEOREM1} is achieved by a set of nontrivial and rather
technically complicated algebraic manipulations. To motivate these
manipulations recall that the $z$-measure $M_{z,z',\xi,\theta=2}$
is manifestly symmetric with respect to $z\longleftrightarrow z'$.
Therefore, the final formula for the correlation function
$\varrho_{n}^{(z,z',\xi,\theta=2)}(x_1,\ldots ,x_n)$ must be
manifestly symmetric with respect to $z\longleftrightarrow z'$ as
well.

It is convenient to introduce three functions on $\Z'\times\Z'$,
namely $I_{z,z',\xi,\theta=2}(x,y)$, $A_{z,z',\xi,\theta=2}(x,y)$,
and $B_{z,z',\xi,\theta=2}(x,y)$.  We will define these functions
in terms of contour integrals. Let $\{w_1\}$ and $\{w_2\}$ be
arbitrary simple contours satisfying the conditions
\begin{itemize}
    \item  both contours go around $0$ in positive direction;
    \item  the point $\xi^{1/2}$ is in the interior of each of the
    contours while the point $\xi^{-1/2}$ lies outside the
    contours;
    \item the contour $\{w_1^{-1}\}$ is contained in the interior
    of the contour $\{w_2\}$ (equivalently, $\{w_2^{-1}\}$ is
    contained in the interior of $\{w_1\}$).
\end{itemize}
If $x-\frac{1}{2}$ is an even integer, and $y-\frac{1}{2}$ is an
arbitrary integer, then the first function,
$I_{z,z',\xi,\theta=2}(x,y)$, is defined by
\begin{equation}\label{Equation(3.2A)}
\begin{split}
I_{z,z',\xi,\theta=2}(x,y)= -\frac{1}{(2\pi
i)^2}\oint\limits_{\{w_1\}}&\oint\limits_{\{w_2\}}\frac{(1-\sqrt{\xi}w_1)^{-z'}(1-\frac{\sqrt{\xi}}{w_1})^{z}
(1-\sqrt{\xi}w_2)^{-z}(1-\frac{\sqrt{\xi}}{w_2})^{z'}}{w_1w_2-1}\\
&\times
F\left(\frac{x+z-\frac{1}{2}}{2},1;\frac{x+z'+\frac{1}{2}}{2};\frac{1}{w_1^2}\right)\frac{dw_1}{w_1^{x-\frac{1}{2}}}\frac{dw_2}{w_2^{y-\frac{1}{2}}},
\end{split}
\end{equation}
where the contours are chosen as described above with an
additional condition that both contours lie in the domain $|w|>1$.\\
If $x-\frac{1}{2}$ is an odd integer, and $y-\frac{1}{2}$ is an
arbitrary integer, then $I_{z,z',\xi,\theta=2}(x,y)$ is defined by
\begin{equation}\label{Equation(3.2B)}
\begin{split}
I_{z,z',\xi,\theta=2}(x,y)=\frac{1}{(2\pi
i)^2}\oint\limits_{\{w_1\}}&\oint\limits_{\{w_2\}}\frac{(1-\sqrt{\xi}w_1)^{-z'}(1-\frac{\sqrt{\xi}}{w_1})^{z}
(1-\sqrt{\xi}w_2)^{-z}(1-\frac{\sqrt{\xi}}{w_2})^{z'}}{w_1w_2-1}\\
&\times\left[F\left(-\frac{x+z'-\frac{3}{2}}{2},1;-\frac{x+z-\frac{5}{2}}{2};w_1^2\right)-1\right]\frac{dw_1}{w_1^{x-\frac{1}{2}}}\frac{dw_2}{w_2^{y-\frac{1}{2}}},
\end{split}
\end{equation}
where the contours are chosen as described above with an
additional condition that both contours lie in the domain $|w|<1$.

Next let us define the second function, namely
$A_{z,z',\xi,\theta=2}(x)$. If $x-\frac{1}{2}$ is an even integer,
then $A_{z,z',\xi,\theta=2}(x)$ is defined by the contour integral
\begin{equation}\label{UIA}
\begin{split}
A_{z,z',\xi,\theta=2}(x)=-\frac{1}{2\pi
i}\oint\limits_{\{w\}}(1-\sqrt{\xi}w)^{-z'-1}(1-\frac{\sqrt{\xi}}{w})^{z}
F\left(\frac{x+z-\frac{1}{2}}{2},1;\frac{x+z'+\frac{1}{2}}{2};\frac{1}{w^2}\right)\frac{dw}{w^{x-\frac{1}{2}}}.
\end{split}
\end{equation}
where $\{w\}$ is an arbitrary simple contour going around $0$ in
the positive direction, and such that it lies in the domain
$|w|>1$.\\
If $x-\frac{1}{2}$ is an odd integer, then we set
\begin{equation}\label{UIB}
\begin{split}
A_{z,z',\xi,\theta=2}(x)=\frac{1}{2\pi
i}&\oint\limits_{\{w\}}(1-\sqrt{\xi}w)^{-z'-1}(1-\frac{\sqrt{\xi}}{w})^{z}\\
&\times\left[F\left(-\frac{x+z'-\frac{3}{2}}{2},1;-\frac{x+z-\frac{5}{2}}{2};w^2\right)-1\right]\frac{dw}{w^{x-\frac{1}{2}}}.
\end{split}
\end{equation}
Here $\{w\}$ is an arbitrary simple contour going around $0$ in
the positive direction, and such that it lies in the domain
$|w|<1$.

 Finally, we define $B_{z,z',\xi,\theta=2}(y)$. This
function has the following contour integral representation. If
$y-\frac{1}{2}$ is an even integer, then
\begin{equation}\label{UIIA}
\begin{split}
B_{z,z',\xi,\theta=2}(y)=\frac{1}{2\pi
i}&\oint\limits_{\{w\}}(1-\sqrt{\xi}w)^{-z}(1-\frac{\sqrt{\xi}}{w})^{z'}\\
&\times\left[1-(1-\frac{\sqrt{\xi}}{w})F\left(\frac{y+z'+\frac{1}{2}}{2},1;\frac{y+z-\frac{1}{2}}{2};\frac{1}{w^2}\right)\right]
\frac{dw}{w^{y-\frac{1}{2}}}.
\end{split}
\end{equation}
Here $\{w\}$ is an arbitrary simple contour going around $0$ in
the positive direction, and such that it lies in the domain
$|w|>1$.\\
 If $y-\frac{1}{2}$ is an odd integer, then
\begin{equation}\label{UIIB}
\begin{split}
B_{z,z',\xi,\theta=2}(y)=\frac{1}{2\pi
i}\oint\limits_{\{w\}}&(1-\sqrt{\xi}w)^{-z}(1-\frac{\sqrt{\xi}}{w})^{z'}\\
&\times\left[(1-\frac{\sqrt{\xi}}{w})F\left(-\frac{y+z-\frac{5}{2}}{2},1;-\frac{y+z'-\frac{3}{2}}{2};w^2\right)+\frac{\sqrt{\xi}}{w}\right]
\frac{dw}{w^{y-\frac{1}{2}}}.
\end{split}
\end{equation}
Here $\{w\}$ is an arbitrary simple contour going around $0$ in
the positive direction, and such that it lies in the domain
$|w|>1$.
\begin{prop}\label{PropositionParticularCase1}
1) For $N=1, 2,\ldots $ let $z=2N$ and $z'=2N+\beta-2$ with
$\beta>0$. Assume that $x_1,\ldots ,x_n$ lie in the subset
$\Zp-2N+\frac{1}{2}\subset\Z'$, so that the points
$\tx_i=x_i+2N-1/2$ are in $\Zp$. Then
$$
\varrho_{n}^{(z,z',\xi,\theta=2)}(x_1,\ldots ,x_n)
=\Pf\left[\mathbb{K}_{z,z',\xi,\theta=2}(x_i,x_j)\right]_{i,j=1}^n,
$$
where the correlation kernel $\mathbb{K}_{z,z',\xi,\theta=2}(x,y)$
has the following form
$$
\mathbb{K}_{z,z',\xi,\theta=2}(x,y)= \left[\begin{array}{cc}
  S_{z,z',\xi,\theta=2}(x,y) & -{SD_-}_{z,z',\xi,\theta=2}(x,y) \\
  -{D_+S}_{z,z',\xi,\theta=2}(x,y) & {D_+SD_-}_{z,z',\xi,\theta=2}(x,y) \\
\end{array}\right].
$$
2) The functions ${SD_-}_{z,z',\xi,\theta=2}(x,y)$,
${D_+S}_{z,z',\xi,\theta=2}(x,y)$, and
${D_+SD_-}_{z,z',\xi,\theta=2}(x,y)$ are expressible in terms of
the function $S_{z,z',\xi,\theta=2}(x,y)$ as follows
$$
{SD_-}_{z,z',\xi,\theta=2}(x,y)=S_{z,z',\xi,\theta=2}(x,y+1)\sqrt{\frac{y+z+\frac{1}{2}}{y+z'+\frac{3}{2}}},
$$
$$
{D_+S}_{z,z',\xi,\theta=2}(x,y)=\sqrt{\frac{x+z+\frac{1}{2}}{x+z'+\frac{3}{2}}}S_{z,z',\xi,\theta=2}(x+1,y),
$$
and
$$
{D_+SD_-}_{z,z',\xi,\theta=2}(x,y)=
\sqrt{\frac{x+z+\frac{1}{2}}{x+z'+\frac{3}{2}}}S_{z,z',\xi,\theta=2}(x+1,y+1)
\sqrt{\frac{y+z+\frac{1}{2}}{y+z'+\frac{3}{2}}}.
$$
3) The function $S_{z,z',\xi,\theta=2}(x,y)$ can be written as
\begin{equation}
\begin{split}
S_{z,z',\xi,\theta=2}(x,y)&=\sqrt{\frac{\Gamma(x+z+\frac{1}{2})
\Gamma(y+z'+\frac{3}{2})}{\Gamma(y+z+\frac{1}{2})\Gamma(x+z'+\frac{3}{2})}}\\
&\times\left[I_{z,z',\xi,\theta=2}(x+2,y+1)
+A_{z,z',\xi,\theta=2}(x+2)B_{z,z',\xi,\theta=2}(y+1)\right],
\end{split}
\nonumber
\end{equation}
and it is related with the kernel $S_{2N}^{\Meixner}(\tx,\ty)$ as
\begin{equation}\label{Equation5.131}
S_{2N}^{\Meixner}(x+2N-\frac{1}{2},y+2N-\frac{1}{2})=\sqrt{\xi}S_{z,z',\xi,\theta=2}(x,y).
\end{equation}
\end{prop}
\begin{proof}
Follows immediately from Theorem \ref{TheoremSN4}, Lemma
\ref{LemmaRelation}, and Proposition \ref{TheoremMeixnerKN4}.
\end{proof}
\begin{prop}\label{PropositionParticularCase2} If $x-\frac{1}{2}$ is an even integer, then the
function $I_{z,z',\xi,\theta=2}(x,y)$ (defined by equation
(\ref{Equation(3.2A)})) can be written as
\begin{equation}
\begin{split}
I_{z,z',\xi,\theta=2}(x,y)&=-\sqrt{\frac{\Gamma(x+z'-\frac{1}{2})\Gamma(y+z-\frac{1}{2})}{\Gamma(x+z-\frac{1}{2})\Gamma(y+z'-\frac{1}{2})}}\\
&\times\sum\limits_{l=0}^{\infty}
\sqrt{\frac{\left(x+z-\frac{1}{2}\right)_{l,2}
\left(x+z'-\frac{1}{2}\right)_{l,2}}{\left(x+z+\frac{1}{2}\right)_{l,2}\left(x+z'+\frac{1}{2}\right)_{l,2}}}
\underline{K}_{z,z',\xi}(x+2l-1,y-1),
\end{split}
\nonumber
\end{equation}
where the function $\underline{K}_{z,z',\xi}(x,y)$ is defined by
equation (\ref{5.6112}). If $x-\frac{1}{2}$ is an odd integer,
then the function $I_{z,z',\xi,\theta=2}(x,y)$ (defined by
equation (\ref{Equation(3.2B)}))  can be written as
\begin{equation}
\begin{split}
I_{z,z',\xi,\theta=2}(x,y)&=\sqrt{\frac{\Gamma(x+z'-\frac{1}{2})\Gamma(y+z-\frac{1}{2})}{\Gamma(x+z-\frac{1}{2})\Gamma(y+z'-\frac{1}{2})}}\\
&\times\sum\limits_{l=1}^{\infty}
\sqrt{\frac{\left(-x-z+\frac{3}{2}\right)_{l,2}
\left(-x-z'+\frac{3}{2}\right)_{l,2}}{\left(-x-z+\frac{5}{2}\right)_{l,2}\left(-x-z'+\frac{5}{2}\right)_{l,2}}}
\underline{K}_{z,z',\xi}(x-2l-1,y-1).
\end{split}
\nonumber
\end{equation}
\end{prop}
\begin{proof}
Rewrite the Gauss hypergeometric functions inside the integrals in
equations (\ref{Equation(3.2A)}), (\ref{Equation(3.2B)}) as
infinite sums. These sums are uniformly convergent in the domains
where the contours of integration are chosen. Therefore we can
interchange summation and integration. The integrals inside the
sums can be expressed in terms of the function
$\underline{K}_{z,z',\xi}(x,y)$ as it follows from Proposition
\ref{EquationRepresentationKzz'}.
\end{proof}
\begin{prop}\label{PropositionParticularCase3} If $x-\frac{1}{2}$ is an even integer, then
\begin{equation}
\begin{split}
A_{z,z',\xi,\theta=2}(x)&=-(1-\xi)^{\frac{z-z'-1}{2}}\sqrt{\frac{\Gamma(z+1)\Gamma(x+z'-\frac{1}{2})}{\Gamma(z'+1)\Gamma(x+z-\frac{1}{2})}}\\
&\times\sum\limits_{l=0}^{\infty}
\sqrt{\frac{\left(z+x-\frac{1}{2}\right)_{l,2}
\left(z'+x-\frac{1}{2}\right)_{l,2}}{\left(z+x+\frac{1}{2}\right)_{l,2}\left(z'+x+\frac{1}{2}\right)_{l,2}}}
\psi_{-\frac{1}{2}}(x+2l-1;z,z',\xi),
\end{split}
\nonumber
\end{equation}
where the function $\psi_{-\frac{1}{2}}(x;z,z',\xi)$ is defined by
equation (\ref{PsiaF}). If $x-\frac{1}{2}$ is an odd integer, then
\begin{equation}
\begin{split}
A_{z,z',\xi,\theta=2}(x)&=(1-\xi)^{\frac{z-z'-1}{2}}\sqrt{\frac{\Gamma(z+1)\Gamma(x+z'-\frac{1}{2})}{\Gamma(z'+1)\Gamma(x+z-\frac{1}{2})}}\\
&\times\sum\limits_{l=1}^{\infty}
\sqrt{\frac{\left(-x-z+\frac{3}{2}\right)_{l,2}
\left(-x-z'+\frac{3}{2}\right)_{l,2}}{\left(-x-z+\frac{5}{2}\right)_{l,2}\left(-x-z'+\frac{5}{2}\right)_{l,2}}}
\psi_{-\frac{1}{2}}(x-2l-1;z,z',\xi).
\end{split}
\nonumber
\end{equation}
\end{prop}
\begin{proof}The function $A_{z,z',\xi,\theta=2}(x)$ is defined by
equations (\ref{UIA}), (\ref{UIIA}). As in the proof of the
previous Proposition represent the Gauss hypergeometric functions
inside the integrals as infinite sums, and interchange summation
and integration. Then use Proposition \ref{PsiaF} to rewrite the
integrals inside the sums in terms of
$\psi_{-\frac{1}{2}}(x;z,z',\xi)$.
\end{proof}
\begin{prop}\label{PropositionParticularCase4}
If $y-\frac{1}{2}$ is an even integer, then the function
$B_{z,z',\xi,\theta=2}(y)$ (defined by equation (\ref{UIIA})) can
be represented as
\begin{equation}
\begin{split}
B_{z,z',\xi,\theta=2}(y)=-(1-\xi)^{\frac{z'-z+1}{2}}\frac{\sqrt{zz'}}{\sqrt{(y+z-\frac{1}{2})(y+z'-\frac{1}{2})}}
\sqrt{\frac{\Gamma(z'+1)\Gamma(y+z-\frac{1}{2})}{\Gamma(z+1)\Gamma(y+z'-\frac{1}{2})}}\\
\times\sum\limits_{l=0}^{\infty}\sqrt{\frac{\left(y+z+\frac{1}{2}\right)_{l,2}\left(y+z'+\frac{1}{2}\right)_{l,2}}{\left(y+z+\frac{3}{2}\right)_{l,2}
\left(y+z'+\frac{3}{2}\right)_{l,2}}}\psi_{1/2}(y+2l;z,z',\xi).
\end{split}
\nonumber
\end{equation}
If $y-\frac{1}{2}$ is an odd integer, then  the function
$B_{z,z',\xi,\theta=2}(y)$ (defined by equation (\ref{UIIB})) can
be represented as
\begin{equation}
\begin{split}
B_{z,z',\xi,\theta=2}(y)=(1-\xi)^{\frac{z'-z+1}{2}}\frac{\sqrt{zz'}}{\sqrt{(y+z-\frac{1}{2})(y+z'-\frac{1}{2})}}
\sqrt{\frac{\Gamma(z'+1)\Gamma(y+z-\frac{1}{2})}{\Gamma(z+1)\Gamma(y+z'-\frac{1}{2})}}\\
\times\sum\limits_{l=1}^{\infty}\sqrt{\frac{\left(-y-z+\frac{1}{2}\right)_{l,2}\left(-y-z'+\frac{1}{2}\right)_{l,2}}{\left(-y-z+\frac{3}{2}\right)_{l,2}
\left(-y-z'+\frac{3}{2}\right)_{l,2}}}\psi_{1/2}(y-2l;z,z',\xi).
\end{split}
\nonumber
\end{equation}
\end{prop}
\begin{proof}
Consider first the case when $y- \frac{1}{2}$ is an even integer.
In this case the function $B_{z,z',\xi,\theta=2}(y)$ is defined by
equation (\ref{UIIA}).  We use the identity
$$
(1-\sqrt{\xi}w)[1-(1-\frac{\sqrt{\xi}}{w})F]=1-(1+\xi)F+\sqrt{\xi}[\frac{F}{w}+w(F-1)]
$$
to rewrite $B_{z,z',\xi,\theta=2}(y)$ as a sum of three terms each
of which is defined by contour integrals. Namely, we have
$$
\widehat{\underline{B}}_{z,z',\xi,\theta=2}(y+\frac{1}{2})=T_1(y)-(1+\xi)T_2(y)+\sqrt{\xi}T_3(y),
$$
where
\begin{equation}
T_1(y)=\frac{1}{2\pi
i}\oint\limits_{\{w\}}(1-\sqrt{\xi}w)^{-z-1}(1-\frac{\sqrt{\xi}}{w})^{z'}
\frac{dw}{w^{y}}, \nonumber
\end{equation}
\begin{equation}
T_2(y)=\frac{1}{2\pi
i}\oint\limits_{\{w\}}(1-\sqrt{\xi}w)^{-z-1}(1-\frac{\sqrt{\xi}}{w})^{z'}
F\left(\frac{y+z'+1}{2},1;\frac{y+z}{2};\frac{1}{w^2}\right)
\frac{dw}{w^{y}}, \nonumber
\end{equation}
and
\begin{equation}
\begin{split}
&T_3(y)=\frac{1}{2\pi
i}\oint\limits_{\{w\}}(1-\sqrt{\xi}w)^{-z-1}(1-\frac{\sqrt{\xi}}{w})^{z'}
F\left(\frac{y+z'+1}{2},1;\frac{y+z}{2};\frac{1}{w^2}\right)
\frac{dw}{w^{y+1}}\\
&+\frac{1}{2\pi
i}\oint\limits_{\{w\}}(1-\sqrt{\xi}w)^{-z-1}(1-\frac{\sqrt{\xi}}{w})^{z'}
\left[F\left(\frac{y+z'+1}{2},1;\frac{y+z}{2};\frac{1}{w^2}\right)-1\right]
\frac{dw}{w^{y-1}}.
\end{split}
\nonumber
\end{equation}
In the formulae for $T_1(y)$, $T_2(y)$, and $T_3(y)$ written above
the contour $\{w\}$ lies in the domain $|w|>1$. We use Proposition
\ref{LemmaFASCountourIntegral} to represent $T_1(y)$ in terms of
the Gauss hypergeometric function, namely we obtain
$$
T_1(y)=(1-\xi)^{z'}\xi^{\frac{y-1}{2}}\frac{\Gamma(z+y)}{\Gamma(z+1)}
\frac{F\left(-z,-z';y;\frac{\xi}{\xi-1}\right)}{\Gamma(y)}.
$$
Consider the expression  for the function $T_2(y)$. We represent
the hypergeometric function inside the integral as the infinite
series,
$$
F\left(\frac{y+z'+1}{2},1;\frac{y+z}{2};\frac{1}{w^2}\right)
=\sum\limits_{l=0}^{\infty}\frac{\left(\frac{y+z'+1}{2}\right)_l}{\left(\frac{y+z}{2}\right)_l}\frac{1}{w^{2l}}.
$$
Again, this series is uniformly convergent on the integration
contour. Therefore, we can interchange the summation and
integration, and  compute the integrals in terms of the Gauss
hypergeometric function using Proposition
\ref{LemmaFASCountourIntegral}. The result is
\begin{equation}
\begin{split}
T_2(y)&=(1-\xi)^{z'}\frac{\Gamma(z+y)}{\Gamma(z+1)}\\
&\times\sum\limits_{l=0}^{\infty} \left(\frac{y+z'+1}{2}\right)_l
\left(\frac{y+z+1}{2}\right)_l
\xi^{\frac{2l+y-1}{2}}2^{2l}\frac{F\left(-z,-z';2l+y;\frac{\xi}{\xi-1}\right)}{\Gamma(2l+y)}.
\end{split}
\nonumber
\end{equation}
In a similar way we find after some calculations that
\begin{equation}
\begin{split}
T_3(y)=(1-\xi)^{z'}\frac{\Gamma(z+y)}{\Gamma(z+1)}\sum\limits_{l=0}^{\infty}
\biggl[&\left(\frac{z'+y+1}{2}\right)_l
\left(\frac{z+y+1}{2}\right)_l\xi^{\frac{2l+y}{2}}2^{2l}\\
&\times(2y+z+z'+4l+1)\frac{F\left(-z,-z';2l+y+1;\frac{\xi}{\xi-1}\right)}{\Gamma(2l+y+1)}\biggr].
\end{split}
\nonumber
\end{equation}
This gives the following expression for
$B_{z,z',\xi,\theta=2}(y+\frac{1}{2})$
\begin{equation}
\begin{split}
&B_{z,z',\xi,\theta=2}(y+\frac{1}{2})=-
(1-\xi)^{z'}\frac{\Gamma(z+y)}{\Gamma(z+1)}
\sum\limits_{l=0}^{\infty} \left(\frac{z+y+1}{2}\right)_l
\left(\frac{z'+y+1}{2}\right)_l\xi^{\frac{2l+y+1}{2}}2^{2l}\\
&\times\biggl[(z+y+1+2l)(z'+y+1+2l)
\frac{F\left(-z,-z';2l+y+2;\frac{\xi}{\xi-1}\right)}{\Gamma(2l+y+2)}+
\frac{
F\left(-z,-z';2l+y;\frac{\xi}{\xi-1}\right)}{\Gamma(2l+y)}\\
&-(2y+z+z'+4l+1)
\frac{F\left(-z,-z';2l+y+1;\frac{\xi}{\xi-1}\right)}{\Gamma(2l+y+1)}
\biggr].
\end{split}
\nonumber
\end{equation}
Set
$$
A=-z,\;\; B=-z',\;\; C=2l+y+1.
$$
Then
$$
2y+z+z'+4l+1=2C-A-B-1,\;\; z+y+1+2l=C-A,\;\; z'+y+1+2l=C-B,
$$
and the sum of three terms in the brackets in the expression for
$B_{z,z',\xi,\theta=2}(y+\frac{1}{2})$ can be rewritten as
\begin{equation}
\begin{split}
(C-A)(C-B)\frac{F(A,B;C+1;w)}{\Gamma(C+1)}+\frac{F(A,B;C-1;w)}{\Gamma(C-1)}
-(2C-A-B-1)\frac{F(A,B;C;w)}{\Gamma(C)}.
\end{split}
\nonumber
\end{equation}
Here $w=\frac{\xi}{\xi-1}$. The following relation for the Gauss
hypergeometric function holds true
\begin{equation}
\begin{split}
&(C-A)(C-B)F(A,B;C+1;w)+C(C-1)F(A,B;C-1;w)\\&-C(2C-A-B-1)F(A,B;C;w)
=ABF(A+1,B+1;C+1;w).
\end{split}
\nonumber
\end{equation}
Using this relation, we rewrite
$B_{z,z',\xi,\theta=2}(y+\frac{1}{2})$ as
\begin{equation}
\begin{split}
B_{z,z',\xi,\theta=2}(y+1/2)&=
-(1-\xi)^{z'}(zz')\frac{\Gamma(z+y)}{\Gamma(z+1)}\\
&\times\sum\limits_{l=0}^{\infty} \left(\frac{z+y+1}{2}\right)_l
\left(\frac{z'+y+1}{2}\right)_l\xi^{\frac{2l+y+1}{2}}2^{2l} \frac{
F\left(-z,-z';2l+y;\frac{\xi}{\xi-1}\right)}{\Gamma(2l+y)}.
\end{split}
\nonumber
\end{equation}
Finally, Proposition \ref{PsiaF} enables us to represent
$B_{z,z',\xi,\theta=2}(y)$ as in the statement of Proposition
\ref{PropositionParticularCase4}. The formula for
$B_{z,z',\xi,\theta=2}(y)$ in the case of an odd $y-\frac{1}{2}$
is obtained by a similar calculation.
\end{proof}
\begin{prop}
For $N=1, 2,\ldots $ let $z=2N$ and $z'=2N+\beta-2$ with
$\beta>0$. Assume that $x_1,\ldots ,x_n$ lie in the subset
$\Zp-2N+\frac{1}{2}\subset\Z'$, so that the points
$\tx_i=x_i+2N-1/2$ are in $\Zp$. Then the correlation function $
\varrho_{n}^{(z,z',\xi,\theta=2)}(x_1,\ldots ,x_n) $ is determined
by the formulae of Theorem \ref{MAINTHEOREM1}.
\end{prop}
\begin{proof}Let the parameters $z, z'$, and the points $x_1,\ldots ,x_n$ be chosen as in the statement of the Proposition.
Then  a straightforward application of Propositions
\ref{PropositionParticularCase1}-\ref{PropositionParticularCase4}
gives the correlation function $
\varrho_{n}^{(z,z',\xi,\theta=2)}(x_1,\ldots ,x_n) $ as a Pfaffian
of a $2\times 2$ block matrix defined by the kernel
$\mathbb{K}_{z,z',\xi,\theta=2}(x,y)$, and the explicit formulae
for the matrix entries of $\mathbb{K}_{z,z',\xi,\theta=2}(x,y)$.
It can be checked that the kernel
$\mathbb{K}_{z,z',\xi,\theta=2}(x,y)$ is equivalent to the kernel
$\underline{\mathbb{K}}_{z,z',\xi,\theta=2}(x,y)$ of Theorem
\ref{MAINTHEOREM1}, i. e.
$$
\Pf\left[\mathbb{K}_{z,z',\xi,\theta=2}(x_i,x_j)\right]_{i,j=1}^n=\Pf\left[\underline{\mathbb{K}}_{z,z',\xi,\theta=2}(x_i,x_j)\right]_{i,j=1}^n,
$$
and that the functions $S_{z,z',\xi,\theta=2}$ and
$\underline{S}_{z,z',\xi,\theta=2}$ that define the matrix kernels
$\mathbb{K}_{z,z',\xi,\theta=2}(x,y)$ and
$\underline{\mathbb{K}}_{z,z',\xi,\theta=2}(x,y)$ are related as
\begin{equation}\label{Equation5.51}
\underline{S}_{z,z',\xi,\theta=2}(x,y)=\sqrt{(x+z+\frac{1}{2})(y+z+\frac{1}{2})}
S_{z,z',\xi,\theta=2}(x,y).
\end{equation}
\end{proof}
\section{Analytic continuation}
\subsection{Analytic properties of correlation functions}
\begin{prop}\label{Proposition4.1}
Fix an arbitrary set of Young diagrams $\mathbb{D}\subset\Y$. For
any fixed admissible pair of parameters $(z,z')$, and for
$\theta>0$, the function
$$
\xi\rightarrow\sum\limits_{\lambda\in\mathbb{D}}M_{z,z',\xi,\theta}(\lambda)
$$
which is initially defined on the interval $(0,1)$ can be extended
to a holomorphic function in the unit disk $|\xi|<1$.
\end{prop}
\begin{proof} Comparing equations (\ref{EquationVer4zmeasuren})
and (\ref{EquationMzztheta}) we obtain
\begin{equation}\label{MzzMzzn}
M_{z,z',\xi,\theta}(\lambda)=(1-\xi)^t\frac{(t)_n}{n!}\xi^n
M_{z,z',\theta}^{(n)}(\lambda),\; n=|\lambda|.
\end{equation}
Set $\mathbb{D}_n=\mathbb{D}\cap\Y_n$. Using equation
(\ref{MzzMzzn}) we can write
$$
\sum\limits_{\lambda\in\mathbb{D}}M_{z,z',\xi,\theta}(\lambda)=(1-\xi)^t\sum\limits_{n=0}^{\infty}
\left(\sum\limits_{\lambda\in\mathbb{D}_n}M_{z,z',\theta}^{(n)}(\lambda)\right)\frac{(t)_n\xi^n}{n!}.
$$
The interior sum is nonnegative, and does not exceed $1$. On the
other hand
$$
\sum\limits_{n=0}^{\infty}\left|\frac{(t)_n\xi^n}{n!}\right|=\sum\limits_{n=0}^{\infty}\frac{(t)_n|\xi|^n}{n!}<\infty,\;
\xi\in\C,\; |\xi|<1.
$$
This shows that the function
$\xi\rightarrow\sum_{\lambda\in\mathbb{D}}M_{z,z',\xi,\theta}(\lambda)$
can be represented as a power series in $\xi$, which is convergent
in the unit disk $|\xi|<1$. Therefore, the function
$\xi\rightarrow
\sum\limits_{\lambda\in\mathbb{D}}M_{z,z',\xi,\theta}(\lambda)$ is
holomorphic in $|\xi|<1$.
\end{proof}
\begin{prop}\label{Proposition4.2}
Fix an arbitrary set of Young diagrams $\mathbb{D}\subset\Y$.
Consider the Taylor expansion of the function
$$
\xi\rightarrow\sum\limits_{\lambda\in\mathbb{D}}M_{z,z',\xi,\theta}(\lambda)
$$
at $\xi=0$,
$$
\sum\limits_{\lambda\in\mathbb{D}}M_{z,z',\xi,\theta}(\lambda)=\sum\limits_{k=0}^{\infty}G^{(\theta)}_{k,\mathbb{D}}(z,z')\xi^k.
$$
Then the coefficients $G_{k,\mathbb{D}}^{(\theta)}(z,z')$ are
polynomial functions in $z,z'$.
\end{prop}
\begin{proof}
By (\ref{EquationMzztheta})
$$
\sum\limits_{\lambda\in\mathbb{D}}M_{z,z',\xi,\theta}(\lambda)=(1-\xi)^{\frac{zz'}{\theta}}\sum\limits_{n=0}^{\infty}\sum\limits_{\lambda\in\mathbb{D}_n}
(z)_{\lambda,\theta}(z')_{\lambda,\theta}\xi^n\frac{1}{H(\lambda,\theta)H'(\lambda,\theta)}.
$$
Recall that $(z)_{\lambda,\theta}$ and $(z')_{\lambda,\theta}$ are
polynomials in variables $z$ and $z'$ correspondingly. We have
$$
(1-\xi)^{\frac{zz'}{\theta}}=\sum\limits_{m=0}^{\infty}\frac{(-\frac{zz'}{\theta})_m\xi^m}{m!}.
$$
Inserting this expansion into the righthand side of the formula
for
$\sum\limits_{\lambda\in\mathbb{D}}M_{z,z',\xi,\theta}(\lambda)$
written above we find
$$
G_{k,\mathbb{D}}^{(\theta)}(z,z')=\sum\limits_{k=0}^{\infty}\sum\limits_{\lambda\in\mathbb{D}_n}
\frac{(-\frac{zz'}{\theta})_{k-n}(z)_{\lambda,\theta}(z')_{\lambda,\theta}}{(k-n)!}\frac{1}{H(\lambda,\theta)H'(\lambda,\theta)}.
$$
Since each $\mathbb{D}_n$ is a finite set, this expression is a
polynomial in variables $z,z'$.
\end{proof}
Set
$$
\underline{K}_{z,z',\xi}(x,y)=\varphi_{z,z'}(x,y)\widehat{\underline{K}}_{z,z',\xi}(x,y),
$$
where
$$
\varphi_{z,z'}(x,y)=\frac{\sqrt{\Gamma(x+z+\frac{1}{2})
\Gamma(x+z'+\frac{1}{2})
\Gamma(y+z+\frac{1}{2})\Gamma(y+z+\frac{1}{2})}}{\Gamma(x+z'+\frac{1}{2})\Gamma(y+z+\frac{1}{2})}.
$$
Then Proposition \ref{EquationRepresentationKzz'} implies that
$\widehat{\underline{K}}_{z,z',\xi}(x,y)$ is representable as a
double contour integral involving elementary functions only.
Namely, we have
\begin{equation}\label{equation6.0}
\begin{split}
\widehat{\underline{K}}_{z,z',\xi}(x,y)=\frac{1}{(2\pi
i)^2}\oint\limits_{\{w_1\}}\oint\limits_{\{w_2\}}&
\frac{(1-\sqrt{\xi}w_1)^{-z'}(1-\frac{\sqrt{\xi}}{w_1})^{z}
(1-\sqrt{\xi}w_2)^{-z}(1-\frac{\sqrt{\xi}}{w_2})^{z'}}{w_1w_2-1}\\
&\times w_1^{-x-\frac{1}{2}}w_2^{-y-\frac{1}{2}}dw_1dw_2,
\end{split}
\end{equation}
where the contours are chosen in the same way as in the statement
of Proposition \ref{EquationRepresentationKzz'}. We also introduce
the functions $\widehat{\psi}_{\frac{1}{2}}(x;z,z',\xi)$ and
$\widehat{\psi}_{-\frac{1}{2}}(x;z,z',\xi)$ that are closely
related to the functions $\psi_{\frac{1}{2}}(x;z,z',\xi)$ and
$\psi_{-\frac{1}{2}}(x;z,z',\xi)$. The functions
$\widehat{\psi}_{\frac{1}{2}}(x;z,z',\xi)$ and
$\widehat{\psi}_{-\frac{1}{2}}(x;z,z',\xi)$ are defined by
equations
$$
\widehat{\psi}_{\frac{1}{2}}(x;z,z',\xi)=f_{z,z',\frac{1}{2}}(x)(1-\xi)^{\frac{z-z'+1}{2}}
\widehat{\psi}_{\frac{1}{2}}(x;z,z',\xi),
$$
and
$$
\widehat{\psi}_{-\frac{1}{2}}(x;z,z',\xi)=f_{z,z',-\frac{1}{2}}(x)(1-\xi)^{\frac{z'-z+1}{2}}
\widehat{\psi}_{-\frac{1}{2}}(x;z,z',\xi).
$$
Here $f_{z,z',\frac{1}{2}}$ and $f_{z,z',-\frac{1}{2}}$ are gamma
prefactors given by
$$
f_{z,z',\frac{1}{2}}(x)=
\frac{\Gamma(z)}{\sqrt{\Gamma(z)\Gamma(z')}}
\frac{\sqrt{\Gamma(x+z+\frac{1}{2})\Gamma(x+z'+\frac{1}{2})}}{\Gamma(x+z+\frac{1}{2})},
$$
and
$$
f_{z,z',-\frac{1}{2}}(x)=
\frac{\Gamma(z'+1)}{\sqrt{\Gamma(z+1)\Gamma(z'+1)}}
\frac{\sqrt{\Gamma(x+z+\frac{1}{2})\Gamma(x+z'+\frac{1}{2})}}{\Gamma(x+z'+\frac{1}{2})}.
$$
Proposition \ref{PSIa} implies that the functions
$\widehat{\psi}_{\frac{1}{2}}(x;z,z',\xi)$ and
$\widehat{\psi}_{-\frac{1}{2}}(x;z,z',\xi)$ have the contour
integral representations
\begin{equation}\label{equation6.01}
\widehat{\psi}_{\frac{1}{2}}(x;z,z',\xi)=\frac{1}{2\pi
i}\oint\limits_{\{w\}}(1-\sqrt{\xi}w)^{-z}(1-\frac{\sqrt{\xi}}{w})^{z'-1}w^{-x-\frac{1}{2}}\frac{dw}{w},
\end{equation}
and
\begin{equation}\label{equation6.02}
\widehat{\psi}_{-\frac{1}{2}}(x;z,z',\xi)=\frac{1}{2\pi
i}\oint\limits_{\{w\}}(1-\sqrt{\xi}w)^{-z'-1}(1-\frac{\sqrt{\xi}}{w})^{z}w^{-x+\frac{1}{2}}\frac{dw}{w}.
\end{equation}
Here the contour $\{w\}$ is chosen as in Proposition \ref{PSIa}.
Note that the functions $\widehat{\psi}_{\frac{1}{2}}(x;z,z',\xi)$
and $\widehat{\psi}_{-\frac{1}{2}}(x;z,z',\xi)$ are defined by
contour integrals involving elementary functions only.

It is convenient to introduce the functions
$\left(\mathcal{E}\widehat{\underline{K}}_{z,z',\xi}\right)(x,y)$,
$\left(\mathcal{E}\widehat{\psi}_{\frac{1}{2}}\right)(x;z,z',\xi)$,
and
$\left(\mathcal{E}\widehat{\psi}_{-\frac{1}{2}}\right)(x;z,z',\xi)$.
If $x-\frac{1}{2}$ is an even integer, then we set
\begin{equation}\label{equation6.1}
\left(\mathcal{E}\widehat{\underline{K}}_{z,z',\xi}\right)(x,y)=-\sum\limits_{l=0}^{\infty}
\frac{(z+x+\frac{3}{2})_{l,2}}{(z'+x+\frac{5}{2})_{l,2}}\widehat{\underline{K}}_{z,z',\xi}(x+2l+1,y),
\end{equation}
\begin{equation}\label{equation6.2}
\left(\mathcal{E}\widehat{\psi}_{\pm\frac{1}{2}}\right)(x;z,z',\xi)=-\sum\limits_{l=0}^{\infty}
\frac{(z+x+\frac{3}{2})_{l,2}}{(z'+x+\frac{5}{2})_{l,2}}\widehat{\psi}_{\pm\frac{1}{2}}(x+2l+1;z,z',\xi).
\end{equation}
If $x-\frac{1}{2}$ is an odd integer, then we set
\begin{equation}\label{equation6.3}
\left(\mathcal{E}\widehat{\underline{K}}_{z,z',\xi}\right)(x,y)=\sum\limits_{l=1}^{\infty}
\frac{(-z'-x-\frac{1}{2})_{l,2}}{(-z-x+\frac{1}{2})_{l,2}}\widehat{\underline{K}}_{z,z',\xi}(x-2l+1,y),
\end{equation}
\begin{equation}\label{equation6.4}
\left(\mathcal{E}\widehat{\psi}_{\pm\frac{1}{2}}\right)(x;z,z',\xi)=\sum\limits_{l=1}^{\infty}
\frac{(-z'-x-\frac{1}{2})_{l,2}}{(-z-x+\frac{1}{2})_{l,2}}\widehat{\psi}_{\pm\frac{1}{2}}(x-2l+1;z,z',\xi).
\end{equation}
Finally, let us introduce the function
$\widehat{\underline{S}}_{z,z',\xi,\theta=2}(x,y)$ by the formula
\begin{equation}\label{equation6.5}
\widehat{\underline{S}}_{z,z',\xi,\theta=2}(x,y)=\left(z+y+\frac{1}{2}\right)\left(\mathcal{E}\widehat{\underline{K}}_{z,z',\xi}\right)(x,y)
+(1-\xi)z'\left(\mathcal{E}\widehat{\psi}_{-\frac{1}{2}}\right)(x;z,z',\xi)\left(\mathcal{E}\widehat{\psi}_{\frac{1}{2}}\right)(y;z,z',\xi).
\end{equation}
\begin{prop}The functions
$\left(\mathcal{E}\widehat{\underline{K}}_{z,z',\xi}\right)(x,y)$,
$\left(\mathcal{E}\widehat{\psi}_{\pm\frac{1}{2}}\right)(x;z,z',\xi)$
and $\widehat{\underline{S}}_{z,z',\xi,\theta=2}(x,y)$ defined by
equations (\ref{equation6.1})-(\ref{equation6.5}) are real
analytic functions of $\sqrt{\xi}$, $0<\sqrt{\xi}<1$, that admit
holomorphic extension to the open unit disk. The Taylor
coefficients of these functions are rational  in variables $z$ and
$z'$.
\end{prop}
\begin{proof} Formulae (\ref{equation6.0}), (\ref{equation6.01}),
(\ref{equation6.02}), and (\ref{equation6.1})-(\ref{equation6.4})
enable us to obtain contour integral representations for the
functions
$\left(\mathcal{E}\widehat{\underline{K}}_{z,z',\xi}\right)(x,y)$
and
$\left(\mathcal{E}\widehat{\psi}_{\pm\frac{1}{2}}\right)(x;z,z',\xi)$.
From these representations, and from equation (\ref{equation6.5})
the statement of the Proposition follows immediately.
\end{proof}
Now we are in position to complete the proof of Theorem
\ref{MAINTHEOREM1}.\\
\subsection{Proof of Theorem \ref{MAINTHEOREM1}}
It was shown (see Section \ref{SECTION5}) that the formula for the
correlation function
$\varrho_n^{(z,z',\xi,\theta=2)}(x_1,\ldots,x_n)$ holds true for
$z=2N$, $z'=2N+\beta-2$, where $N=1,2,\ldots,$ and $\beta>0$. We
want to extend this formula for all admissible values of
parameters $(z,z')$. Assume that $z=2N$ and $z'=2N+\beta-2$. Then
a straightforward algebra (and the fact that
$\underline{S}_{z,z',\xi}(x,y)$ is antisymmetric for $z=2N$ and
$z'=2N+\beta-2$) gives the following expressions for the matrix
elements of $\underline{\mathbb{K}}_{z,z',\xi}(x,y)$:
$$
\underline{S}_{z,z',\xi,\theta=2}(x,y)=\varphi_{z,z'}(x+1,y+1)\widehat{\underline{S}}_{z,z',\xi,\theta=2}(x,y),
$$
$$
\underline{D_+S}_{z,z',\xi,\theta=2}(x,y)=\frac{\varphi_{z,z'}(x+1,y+1)}{z'+x+\frac{3}{2}}\widehat{\underline{S}}_{z,z',\xi,\theta=2}(x+1,y),
$$
$$
\underline{SD_-}_{z,z',\xi,\theta=2}(x,y)=-\frac{1}{(z'+y+\frac{3}{2})\varphi_{z,z'}(x+1,y+1)}
\widehat{\underline{S}}_{z,z',\xi,\theta=2}(y+1,x),
$$
and
$$
\underline{D_+SD_-}_{z,z',\xi,\theta=2}(x,y)=-\frac{1}{(z'+x+\frac{3}{2})(z'+y+\frac{3}{2})\varphi_{z,z'}(x+1,y+1)}
\widehat{\underline{S}}_{z,z',\xi,\theta=2}(y+1,x+1).
$$
Computing the Pfaffian in the righthand side of equation
(\ref{equation2.111}) we see that the function
$\varphi_{z,z'}(x,y)$ (which is the gamma prefactor) is completely
cancelled out. Therefore, the righthand side of equation
(\ref{equation2.111}) has the same property as the function
$\widehat{\underline{S}}_{z,z',\xi,\theta=2}(x,y)$: it is a real
analytic function in $\sqrt{\xi}$, $0<\sqrt{\xi}<1$, that admits a
holomorphic extension to the open unit disk. Moreover, the Taylor
coefficients of this function are rational in $z, z'$. On the
other hand, Propositions \ref{Proposition4.1},
\ref{Proposition4.2} imply that the left-hand side of equation
(\ref{equation2.111}) has the same property, with $\sqrt{\xi}$
replaced by $\xi$. Thus, both sides of equation
(\ref{equation2.111}) can be viewed as holomorphic functions with
 Taylor coefficients rational in $z$ and $z'$. Since the set
 $$
 \{(z,z'): z\;\; \mbox{is  a large natural number}\;\; 2N\;\;
 \mbox{and}\;\;
 z'>2N-2\}
 $$
is a set of uniqueness of rational functions in two variables $z,
z'$, we conclude that  equation (\ref{equation2.111}) holds true
for all admissible $z, z'$. \qed
\section{Proof of Propositions \ref{Proposition2.12} and \ref{Proposition2.13}}
\subsection{Proof of Proposition\ref{Proposition2.12}}
For $N=1,2,\ldots$, let $z=2N$ and
$z'=2N+\beta-2$ with $\beta>0$. Assume that $x, y$ lie in the
subset $\Zp-2N+\frac{1}{2}\in\Z'$. Then the formula for
$\underline{S}_{z,z',\xi,\theta=2}(x,y)$ is obtained form
Proposition \ref{Proposition4.2}, and equations
(\ref{Equation5.131}), (\ref{Equation5.51}). Thus the Proposition
is proved for these specific values of the parameters $z$ and
$z'$. Now we claim that the expression in the righthand side of
the formula for $\underline{S}_{z,z',\xi,\theta=2}(x,y)$ in the
statement of the Proposition is identically equal to the
expression in the righthand side of formula for
$\underline{S}_{z,z',\xi,\theta=2}(x,y)$ in the statement of
Theorem \ref{MAINTHEOREM1}. Indeed, using the identity
\begin{equation}
\begin{split}
&\frac{1}{(2\pi
i)^2}\sqrt{\frac{\Gamma(x+z+\frac{1}{2})\Gamma(y+z'+\frac{1}{2})}{\Gamma(x+z'+\frac{1}{2})\Gamma(y+z+\frac{1}{2})}}
\\
&\times \oint\limits_{\{w_1\}}\oint\limits_{\{w_2\}}
\frac{(1-\sqrt{\xi}w_1)^{-z'}(1-\frac{\sqrt{\xi}}{w_1})^{z}
(1-\sqrt{\xi}w_2)^{-z}(1-\frac{\sqrt{\xi}}{w_2})^{z'}}{w_1w_2-1}
\frac{dw_1}{w_1^{x+\frac{1}{2}}}\frac{dw_2}{w_2^{y+\frac{1}{2}}}\\
&=\frac{1}{(2\pi
i)^2}\sqrt{\frac{\Gamma(y+z+\frac{1}{2})\Gamma(x+z'+\frac{1}{2})}{\Gamma(y+z'+\frac{1}{2})\Gamma(x+z+\frac{1}{2})}}\\
&\times \oint\limits_{\{w_1\}}\oint\limits_{\{w_2\}}
\frac{(1-\sqrt{\xi}w_1)^{-z'}(1-\frac{\sqrt{\xi}}{w_1})^{z}
(1-\sqrt{\xi}w_2)^{-z}(1-\frac{\sqrt{\xi}}{w_2})^{z'}}{w_1w_2-1}
\frac{dw_1}{w_1^{y+\frac{1}{2}}}\frac{dw_2}{w_2^{x+\frac{1}{2}}}
\end{split}
\nonumber
\end{equation}
(which follows from Proposition \ref{EquationRepresentationKzz'},
and from the fact that
$\underline{K}_{z,z',\xi}(x,y)=\underline{K}_{z,z',\xi}(y,x)$) we
can rewrite the righthand side of equation (\ref{Equation2.121})
as a double contour integral. In this way we arrive to formula
(\ref{Equation5.51}) for $\underline{S}_{z,z',\xi,\theta=2}(x,y)$,
where the function  $S_{z,z',\xi,\theta=2}(x,y)$ is given in
Proposition \ref{PropositionParticularCase1}, 3). The expressions
for $\underline{S}_{z,z',\xi,\theta=2}(x,y)$ and
$S_{z,z',\xi,\theta=2}(x,y)$ hold now for all admissible $z$ and
$z'$. Repeating the algebraic calculations as in Propositions
\ref{PropositionParticularCase2}-\ref{PropositionParticularCase4}
we obtain the formula for $S_{z,z',\xi,\theta=2}(x,y)$ stated in
Theorem \ref{MAINTHEOREM1}. \qed
\subsection{Proof of Proposition \ref{Proposition2.13}}
Using formulae of Theorem \ref{MAINTHEOREM1}, we check by direct
calculations that
$\varrho_n^{(z,z'-1,\xi,\theta=2)}(x_1,\ldots,x_n)$ can be written
as in the statement of  Proposition  \ref{Proposition2.13}. In
particular, $\varrho_n^{(z,z'-1,\xi,\theta=2)}(x_1,\ldots,x_n)$ is
determined by the kernel
$\widehat{\underline{S}}_{z,z'=z-1,\xi,\theta=2}(x,y)$ defined by
equation (\ref{ContorS4z'=z-1}). However, the contours $\{w_1\}$
and $\{w_2\}$ in equation (\ref{ContorS4z'=z-1}) still must be
chosen  according to the parities of $x-\frac{1}{2}$ and
$y-\frac{1}{2}$. In particular, we can choose $\{w_1\}$, $\{w_2\}$
to be circular contours such that $|w_1|>1$ in the case when
$x-\frac{1}{2}$ is even, $|w_1|<1$ in the case when
$x-\frac{1}{2}$ is odd, $|w_2|>1$ in the case of an even
$y-\frac{1}{2}$, and $|w_2|<1$ in the case of an odd
$y-\frac{1}{2}$. In addition, we require that all three conditions
in Proposition \ref{EquationRepresentationKzz'} on the contours
$\{w_1\}$, $\{w_2\}$ are satisfied. In particular, in the case
when both $x-\frac{1}{2}$, $y-\frac{1}{2}$ are even,  we choose
$|w_1|>1$ and $|w_2|>1$, so the statement of the Proposition holds
true for even $x-\frac{1}{2}$ and $y-\frac{1}{2}$.

Now we are going to show that  equation (\ref{ContorS4z'=z-1})
with $|w_1|>1$ and $|w_2|>1$ holds true no matter what parities of
$x-\frac{1}{2}$ and $y-\frac{1}{2}$ are. Assume, for example, that
$x-\frac{1}{2}$ is even, and $y-\frac{1}{2}$ is odd. Then the
contours $\{w_1\}$ and $\{w_2\}$ must be chosen in formula
(\ref{ContorS4z'=z-1}) such that $|w_1|>1$, and $|w_2|<1$. Let us
transform the integral in the righthand side of equation
(\ref{ContorS4z'=z-1}): keeping the contour $\{w_1\}$ unchanged we
move $\{w_2\}$ outside the circle of the radius $1$. As a result
contributions from the residues of the function
$w_2\rightarrow\frac{1}{w_2^2-1}$ will arise. These contributions
are
\begin{equation}
\begin{split}
&\frac{1}{2\pi
i}\oint\limits_{\{w_1\}}\frac{(1-\sqrt{\xi}w_1)^{-z}(1-\frac{\sqrt{\xi}}{w_1})^z}{w_1-1}
\frac{(1-w_1)}{2(w_1^2-1)}\frac{dw_1}{w_1^{x-1/2}}\\
&+\frac{(-1)^{y-1/2}}{2\pi
i}\oint\limits_{\{w_1\}}\frac{(1-\sqrt{\xi}w_1)^{-z}(1-\frac{\sqrt{\xi}}{w_1})^z}{-w_1-1}
\frac{(-1-w_1)}{-2(w_1^2-1)}\frac{dw_1}{w_1^{x-1/2}}=0,
\end{split}
\nonumber
\end{equation}
since $y-1/2$ is odd. Other cases can be considered in the same
way. \qed

\end{document}